\documentclass[runningheads]{llncs}
\usepackage[T1]{fontenc}
\usepackage{graphicx}
\usepackage{wrapfig}
\usepackage{xcolor}
\usepackage[ruled,vlined,linesnumbered]{algorithm2e}

\usepackage{amsmath}
\usepackage{amssymb}
\usepackage{commath}
\usepackage[font=small]{caption}
\usepackage{subcaption}
\usepackage{enumerate}
\usepackage{pifont}
\usepackage{amsfonts}
\usepackage{multirow}
\usepackage{makecell}
\usepackage{booktabs} 
\usepackage{makecell}
\usepackage{dsfont}
\usepackage{bm}
\usepackage{thm-restate}
\usepackage[ruled,vlined,linesnumbered]{algorithm2e}
\usepackage{bbding}
\usepackage{tikz}
\usepackage{siunitx}

\usetikzlibrary{angles,quotes,calc}
\SetKwRepeat{Do}{do}{while}
\newcommand{\nn}{f}
\newcommand{\nnunder}{\underline{f}}
\newcommand{\indim}{d}
\newcommand{\outdim}{m}
\newcommand{\inpoint}{x}
\newcommand{\outpoint}{y}

\newcommand{\indomain}{\mathcal{C}}

\newcommand{\intervar}{z^{(i)}_{j}}

\newcommand{\spec}{g}
\newcommand{\net}{f}
\newcommand{\specnum}{K}

\newcommand{\preimage}{\nn^{-1}}
\newcommand{\polytopeset}{\mathcal{T}}
\newcommand{\polytope}{T}
\newcommand{\polytopenum}{D}

\newcommand{\threshold}{v}
\newcommand{\inset}{I}
\newcommand{\outset}{O}
\newcommand{\proportion}{p}

\newcommand{\numsamples}{N}
\newcommand{\nnslope}{\bm{\alpha}}
\newcommand{\nnslopesingle}{\alpha}

\newcommand{\weight}{\textbf{W}}
\newcommand{\bias}{\textbf{b}}
\newcommand{\preact}{h}
\newcommand{\postact}{a}
\newcommand{\numlayers}{L}
\newcommand{\relu}{\textnormal{ReLU}}

\newcommand{\lowerweight}{\underline{\textbf{A}}}
\newcommand{\lowerbias}{\underline{\textbf{b}}}
\newcommand{\lowerweightsingle}{\underline{a}}
\newcommand{\lowerbiassingle}{\underline{b}}
\newcommand{\upperweight}{\overline{\textbf{A}}}
\newcommand{\upperbias}{\overline{\textbf{b}}}
\newcommand{\upperweightsingle}{\overline{a}}
\newcommand{\upperbiassingle}{\overline{b}}
\newcommand{\concretelower}{\textbf{l}}
\newcommand{\concreteupper}{\textbf{u}}
\newcommand{\concretelowersingle}{l}
\newcommand{\concreteuppersingle}{u}

\newcommand{\featurelowerbound}{\underline{x}}
\newcommand{\featureupperbound}{\overline{x}}

\newcommand{\lincon}{\psi}
\newcommand{\boxcon}{\phi}
\newcommand{\fol}{A}
\newcommand{\folunder}{\alpha}
\newcommand{\folover}{\alpha}

\newcommand{\volume}{\textnormal{vol}}
\newcommand{\cov}{\textnormal{cov}}

\newcommand{\algo}{QV}

\newcommand{\linconw}{c}
\newcommand{\linconb}{d}

\newcommand{\maxiterations}{R}
\newcommand{\targetcov}{r}

\begin{document}

\title{Provable Preimage Under-Approximation for Neural Networks}

\titlerunning{Provable Preimage Under-Approximation for Neural Networks}

\author{Xiyue Zhang\inst{(}\Envelope\inst{)}  \and
Benjie Wang\inst{} \and
Marta Kwiatkowska\inst{}}

\authorrunning{X. Zhang et al.}
%
\institute{Department of Computer Science, University of Oxford, Oxford, UK \\
\email{\{xiyue.zhang, benjie.wang, marta.kwiatkowska\}@cs.ox.ac.uk}}
\maketitle              

\begin{abstract}
Neural network verification mainly focuses on local robustness properties, which can be checked by bounding the image (set of outputs) of a given input set. 
However, often it is important to know whether a given property holds globally for the input domain, and if not then for what proportion of the input the property is true. 
To analyze such properties requires computing \emph{preimage} abstractions of neural networks.
In this work, we propose an efficient anytime algorithm for generating symbolic under-approximations of the preimage of any polyhedron output set for neural networks.
Our algorithm combines a novel technique for cheaply computing polytope preimage under-approximations using linear relaxation, with a carefully-designed refinement procedure that iteratively partitions the input region into subregions using input and ReLU splitting in order to improve the approximation.
Empirically, we validate the efficacy of our method across a range of domains, including a high-dimensional MNIST classification task beyond the reach of existing preimage computation methods. Finally, as use cases, we showcase the application to quantitative verification and robustness analysis. We present a sound and complete algorithm for the former, which exploits our disjoint union of polytopes representation to provide formal guarantees. For the latter, we find that our method can provide useful quantitative information even when standard verifiers cannot verify a robustness property. 
\end{abstract}

\section{Introduction}
\label{sec:introduction}
Despite the remarkable empirical success of neural networks, guaranteeing their correctness, especially when using them as decision-making components in safety-critical autonomous systems~\cite{Bojarski16AutoControl,Codevilla18,Yun17Robotics}, is an important and challenging task.
Towards this aim, various approaches have been developed for the verification of neural networks, with extensive effort devoted to local robustness verification \cite{huang2017safety,katz2017reluplex,zhang2018crown,bunel2018unified,tjeng2019evaluating,singh2019deeppoly,xu2020automated,xu2021fast,wang2021beta}.
While local robustness verification focuses on deciding 
the absence of adversarial examples within an $\epsilon$-perturbation neighbourhood, an alternative approach for neural network analysis is to construct the preimage of its predictions \cite{Matoba20Exact,Dathathri19Inverse}. 
Given a set of outputs, the preimage is defined as the set of all inputs mapped by the neural network to that output set. By characterizing the preimage 
symbolically in an abstract representation, e.g., polyhedra, one can perform more complex analysis for a wider class of properties beyond local robustness, such as computing the \emph{proportion} of inputs satisfying a property (quantitative verification) even if standard robustness verification fails.

Exact preimage generation~\cite{Matoba20Exact} is intractable, taking time exponential in the number of neurons in a network; thus approximations are necessary.
Unfortunately, existing methods are limited in their applicability.
The inverse abstraction method in \cite{Dathathri19Inverse}
bypasses the intractability of exact preimage generation by leveraging symbolic interpolants \cite{craig1957inter,Albarghouthi13interpolants} for abstraction of neural network layers.
However, due to the complexity of interpolation, the time
to compute the abstraction also scales exponentially with the number of neurons in hidden layers.
A concurrent work \cite{ProveBound} proposed an input bounding algorithm targeting backward reachability analysis for control policies and out-of-distribution (OOD) detection in low-dimensional domains. Their method produces a preimage \emph{over-approximation}, which cannot be used for quantitative verification.
Therefore, more efficient and flexible computation methods for (symbolic abstraction of) preimages of neural networks are needed.

The main contribution of this paper is a scalable method for preimage approximation, which can be used for a variety of robustness analysis tasks.  
More specifically, we propose an efficient \textit{anytime} algorithm for generating symbolic under-approximations of the preimage of piecewise linear neural networks as a union of disjoint polytopes. 
The algorithm computes a sound preimage under-approximation leveraging linear relaxation based perturbation analysis (LiRPA) \cite{xu2020automated,xu2021fast,singh2019deeppoly}, applied backwards from a polyhedron output set. 
It iteratively refines the preimage approximation by adding input and/or intermediate (ReLU) splitting (hyper)planes to partition the input region into disjoint subregions, which can be approximated independently in parallel in a divide-and-conquer approach. 
The refinement scheme uses a novel differential objective to optimize the quality (volume) of the polytope subregions.
We also show that our method can be generalized to generate preimage over-approximations.
We illustrate the application of our method to quantitative verification,
input bounding for control tasks, and robustness analysis against adversarial and patch attacks.  
Finally, we conduct an empirical analysis 
on a range of control and computer vision tasks, showing significant gains in efficiency compared to exact preimage generation methods and scalability to high-input-dimensional tasks compared to existing preimage approximation methods.

Proofs and additional technical details are presented in the Appendix.

\section{Preliminaries}
\label{sec:pre}
We use 
$\net: \mathbb{R}^{\indim} \to \mathbb{R}^{\outdim}$
to denote a feedforward neural network. For layer $i$, we use $\weight^{(i)}$ to denote the weight matrix, $\bias^{(i)}$ 
the bias, $\preact^{(i)}$ 
the pre-activation neurons, and $\postact^{(i)}$ 
the post-activation neurons, such that we have $\preact^{(i)} = \weight^{(i)} \postact^{(i-1)} + \bias^{(i)}$. 
In this paper, we focus on ReLU neural networks with $\postact^{(i)}(\inpoint)=ReLU(\preact^{(i)}(\inpoint))$, where $\relu(\preact) := \max(\preact, 0)$ is applied element-wise.
However, our method can be generalized to other activation functions bounded by linear relaxation~\cite{zhang2018crown}.

\textbf{Linear Relaxation of Neural Networks.}
Nonlinear activation functions lead to the NP-completeness of the neural network verification problem \cite{katz2017reluplex}. 
To address such intractability, linear relaxation is often used to transform the nonconvex constraints into linear programs. 
As shown in Figure~\ref{fig:linear_relaxation}, given  \textit{concrete}  lower and upper bounds $\concretelower^{(i)}\leq \preact^{(i)}(\inpoint) \leq \concreteupper^{(i)}$ on the pre-activation values of layer $i$, there are three cases to consider.
In the \emph{inactive} ($\concreteuppersingle^{(i)}_j \leq 0$) and \emph{active} ($\concretelowersingle^{(i)}_j \ge 0)$ cases, the post-activation neurons $\postact^{(i)}_j(\inpoint)$ are linear functions $\postact^{(i)}_j(\inpoint) = 0$ and $\postact^{(i)}_j(\inpoint) = \preact^{(i)}_j(\inpoint)$ respectively.
In the \emph{unstable} case, $\postact^{(i)}_j(\inpoint)$ can be bounded by 
$
\nnslopesingle^{(i)}_j \preact^{(i)}_j(\inpoint) \leq \postact^{(i)}_j(\inpoint) \leq 
-\frac{\concreteuppersingle^{(i)}_j\concretelowersingle^{(i)}_j}{\concreteuppersingle^{(i)}_j - \concretelowersingle^{(i)}_j}
    + \frac{\concreteuppersingle^{(i)}_j}{\concreteuppersingle^{(i )}_j - \concretelowersingle^{(i )}_j} \preact^{(i)}_j(\inpoint)
$,
where $\alpha^{(i)}_j$ is a configurable parameter that produces a valid lower bound for any value in 
$[0,1]$. 
Linear bounds can also be obtained for other non-piecewise linear activation functions~\cite{zhang2018crown}.

\begin{figure}[t]
     \centering
     \begin{subfigure}[b]{0.2\textwidth}
         \centering
         \includegraphics[width=\textwidth]{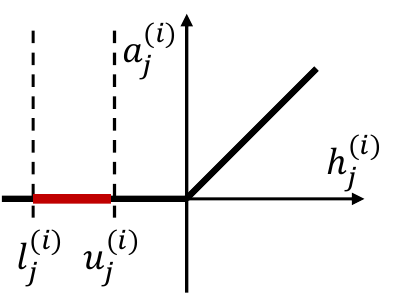}
         \label{fig:negative}
     \end{subfigure}
     \quad
     \begin{subfigure}[b]{0.2\textwidth}
         \centering
         \includegraphics[width=\textwidth]{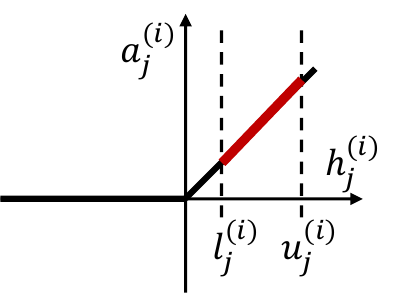}
         \label{fig:positive}
     \end{subfigure}
     \quad
     \begin{subfigure}[b]{0.2\textwidth}
         \centering
         \includegraphics[width=\textwidth]{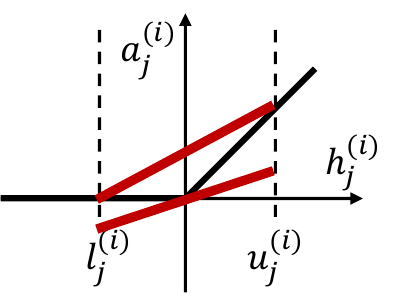}
         \label{fig:unstable}
     \end{subfigure}
        \caption{Linear bounding functions for inactive, active, unstable ReLU neurons.}
        \label{fig:linear_relaxation}
\end{figure}

Linear relaxation can be used to compute linear upper and lower bounds of the form $\lowerweight \inpoint + \lowerbias \leq f(\inpoint) \leq \upperweight \inpoint + \upperbias$ on the output of a neural network, for a given bounded input region $\indomain$.
These methods are known as linear relaxation based perturbation analysis (LiRPA) algorithms \cite{xu2020automated,xu2021fast,singh2019deeppoly}. In particular,
\emph{backward-mode} LiRPA computes linear bounds on 
$f$
by propagating linear 
bounding functions 
backward from the output, layer-by-layer, to the input layer.

\textbf{Polytope Representations.}
Given an Euclidean space $\mathbb{R}^{\indim}$, a polyhedron $\polytope$ is defined to be the intersection of a set of half spaces. More formally, suppose we have a set of linear constraints defined by $\lincon_i(\inpoint) := \linconw_i^T \inpoint + \linconb_i \geq 0$ for $i = 1, ... \specnum$, where $\linconw_i \in \mathbb{R}^{\indim}, \linconb_i \in \mathbb{R}$ are constants, and $\inpoint = \inpoint_1, ..., \inpoint_\indim$ is a set of variables. Then a polyhedron is defined as $\polytope = \{\inpoint \in \mathbb{R}^{\indim}| \bigwedge_{i = 1}^{\specnum} \lincon_i(\inpoint) \}$,
where $\polytope$ consists of all values of $\inpoint$ satisfying the first-order logic (FOL) formula $\folunder(\inpoint) := \bigwedge_{i = 1}^{\specnum} \lincon_i(\inpoint)$.
We use the term polytope to refer to a bounded polyhedron, that is, a polyhedron $\polytope$ such that $\exists R \in \mathbb{R}^{> 0} : \forall \inpoint_1, \inpoint_2 \in \polytope$, $\norm{\inpoint_1 - \inpoint_2}_2 \le R$ holds. 
The abstract domain of polyhedra \cite{singh2019deeppoly,Benoy02Polyhedral,Boutonnet19Polyhedron} has been widely used for the verification of neural networks and computer programs. 
An important type of polytope is the hyperrectangle (box), which is a polytope defined by a closed and bounded interval $[\underline{\inpoint_i}, \overline{\inpoint_i}]$ for each dimension, where $\underline{\inpoint_i}, \overline{\inpoint_i} \in \mathbb{Q}$. More formally, using the linear constraints $\boxcon_i := (\inpoint_i \geq \underline{\inpoint_i}) \wedge (\inpoint_i \leq \overline{\inpoint_i})$ for each dimension, the hyperrectangle takes the form $\indomain = \{\inpoint \in \mathbb{R}^{\indim} | \inpoint \models \bigwedge_{i = 1}^{\indim} \boxcon_i \}$.

\section{Problem Formulation}
\label{sec:problem_definition}

\subsection{Preimage Approximation}

In this work, we are interested in the problem of computing preimages for neural networks. Given a subset $\outset \subset \mathbb{R}^{\outdim}$ of the codomain, the preimage of a function $\nn: \mathbb{R}^{\indim} \to \mathbb{R}^{\outdim}$ is defined to be the set of all inputs $\inpoint \in \mathbb{R}^{\indim}$ that are mapped to an element of $\outset$ by $\nn$. For neural networks in particular, the input is typically restricted to some bounded input region $\indomain \subset \mathbb{R}^{\indim}$. In this work, we restrict the output set $\outset$ to be a polyhedron, and the input set $\indomain$ to be an axis-aligned hyperrectangle region $\indomain \subset \mathbb{R}^{\indim}$, as these are commonly used in neural network verification. We now define the notion of a restricted preimage:

\begin{definition}[Restricted Preimage]
Given a neural network $\nn: \mathbb{R}^{\indim} \to \mathbb{R}^{\outdim}$, and an input set $\indomain \subset \mathbb{R}^{\indim}$, the restricted preimage of an output set $\outset \subset \mathbb{R}^{\outdim}$ is defined to be the set $\nn^{-1}_{\indomain}(\outset) := \{\inpoint \in \mathbb{R}^{\indim}| \nn(\inpoint) \in \outset \wedge \inpoint \in \indomain\}$.
\end{definition}

\begin{example}\label{eg:formulation}
To illustrate our problem formulation and approach,
we introduce a vehicle parking task~\cite{Ayala11vehicle} as a running example.  
In this task, there are four parking lots, located in each quadrant of a $2\times 2$ grid $[0,2]^2$, and a neural network with two hidden layers of 10 ReLU neurons $\nn: \mathbb{R}^2 \to \mathbb{R}^4$ is trained to classify which parking lot an input point belongs to.
To analyze the behaviour of the neural network in the input region $[0, 1] \times [0, 1]$ corresponding to parking lot 1, we set $\indomain = \{\inpoint \in \mathbb{R}^2 | (0 \leq \inpoint_1 \leq 1) \wedge (0 \leq \inpoint_2 \leq 1)\}$. Then the restricted preimage $\nn^{-1}_{\indomain}(\outset)$ of the set $\outset = \{\bm{\outpoint} \in \mathbb{R}^4 | \bigwedge_{i \in \{2, 3, 4\}} \outpoint_1 - \outpoint_i \geq 0\} $ is the subspace of the region $[0, 1] \times [0, 1]$ that is \textit{labelled} as parking lot $1$ by the network. 
\end{example}

We focus on \emph{provable} approximations of the preimage. Given a first-order formula $\fol$, $\folunder$ is an \emph{under-approximation} (resp. \emph{over-approximation}) of $\fol$ if it holds that $\forall \inpoint. \folunder(\inpoint) \implies \fol(\inpoint)$ (resp. $\forall \inpoint. \fol(\inpoint) \implies \folover(\inpoint)$). 
In our context, the restricted preimage is defined by the formula $\fol(\inpoint) = (\nn(\inpoint) \in \outset) \wedge (\inpoint \in \indomain)$, and we restrict to approximations $\folunder$ that take the form of a disjoint union of polytopes (DUP). 
The goal of our method is to generate a DUP approximation $\polytopeset$ that is as tight as possible; that is, to maximize the volume $\volume(\polytopeset)$ of an under-approximation, or minimize the volume $\volume(\polytopeset)$ of an over-approximation.

\begin{definition}[Disjoint Union of Polytopes]
A disjoint union of polytopes (DUP) is a FOL formula $\folunder$ of the form
$    \folunder(\inpoint) := \bigvee_{i = 1}^{\polytopenum} \folunder_i(\inpoint)$, 
where each $\folunder_i$ is a polytope formula 
(conjunction of a finite set of linear half-space constraints), with the property that $\folunder_i \wedge \folunder_j$ is unsatisfiable for any $i \neq j$.
\end{definition}

\subsection{Quantitative Properties} 

One of the most important verification problems for neural networks is that of proving guarantees on the output of a network for a given input set \cite{gehr2018ai2,gopinath2019properties,ruan2018reachability}. This is often expressed as a property of the form $(\inset, \outset)$ such that $\forall \inpoint \in \inset \implies \nn(\inpoint) \in \outset$. We can generalize this to \emph{quantitative} properties:

\begin{definition}[Quantitative Property]
    Given a neural network $\nn: \mathbb{R}^{\indim} \to \mathbb{R}^{\outdim}$, a measurable input set with non-zero measure (volume) $\inset \subseteq \mathbb{R}^{\indim}$, a measurable output set $\outset \subseteq \mathbb{R}^{\outdim}$, and a rational proportion $\proportion \in [0, 1]$ we say that the neural network satisfies the property $(\inset, \outset, \proportion)$ if $\frac{\volume(\nn^{-1}_{\inset}(\outset))}{\volume(\inset)} \geq \proportion$. \footnote{In particular, the restricted preimage of a polyhedron under a neural network is Lebesgue measurable since polyhedra 
    (intersection of a finite set of half-spaces) are Borel measurable and NNs are continuous functions.}
\end{definition}

Neural network verification algorithms~\cite{liu2021algorithms} can be divided into two categories: sound, which always return correct results, and complete, guaranteed to reach a conclusion on any verification query.
We now define soundness and completness of verification algorithms for quantitative properties.

\begin{definition}[Soundness]
    A verification algorithm $\algo$ is sound if, whenever $\algo$ outputs \textnormal{True}, the property $(\inset, \outset, \proportion)$ holds.
\end{definition}

\begin{definition}[Completeness]
    A verification algorithm $\algo$ is complete if (i) $\algo$ never returns \textnormal{Unknown}, and (ii) whenever $\algo$ outputs \textnormal{False}, the property $(\inset, \outset, \proportion)$ does not hold. 
\end{definition}

If the property $(\inset, \outset)$ holds, then the quantitative property $(\inset, \outset, 1)$ holds, while quantitative properties for $0 \leq  \proportion < 1$ provide more information when $(\inset, \outset)$ does not hold. 
Most neural network verification methods produce approximations of the \emph{image} of $\inset$ in the output space, which cannot be used to verify quantitative properties. 
Preimage \textit{over-approximations} include false regions, thereby not applicable for quantitative verification.
In contrast, preimage \emph{under-approximations} provide a lower bound on the volume of the preimage, allowing us to soundly verify quantitative properties.

\section{Methodology}
\label{sec:method}

\subsubsection{Overview.} 
In this section we present the main components of our methodology.  
Firstly, in Section \ref{sec:poly_gen}, we show how to cheaply and soundly under-approximate the (restricted) preimage with a single polytope, using linear relaxation methods (Algorithm \ref{alg:genunderapprox}). Secondly, in Section \ref{sec:local_opt}, we propose a novel differentiable objective to optimize the quality (volume) of the polytope under-approximation. Thirdly, in Section \ref{sec:branching}, we propose a refinement scheme that improves the approximation by partitioning a (sub)region into subregions with splitting planes, with each subregion then being under-approximated more accurately.
The main contribution of this paper 
(Algorithm \ref{alg:main}) integrates these three components and is described in Section \ref{sec:overall}. 
Finally, in Section \ref{sec:verif}, we apply our method to quantitative verification (Algorithm \ref{alg:verify}) and prove its soundness and completeness.

\begin{algorithm}[htb]
\small
\caption{Preimage Approximation}\label{alg:main}
\KwIn{Neural network $f$, Input region $\indomain$, Output region $\outset$, Volume threshold $\threshold$, Maximum iterations $\maxiterations$, Boolean $SplitOnInput$ }
\KwOut{Disjoint union of polytopes $\polytopeset$}
$\polytope$ $\leftarrow$ GenUnderApprox($\indomain$, $\outset$) \tcp*{Initial preimage polytope}\label{algline:initial_polytope}
$\widehat{\volume}_{\polytope}, \widehat{\volume}_{\preimage_{\indomain}(\outset)} \leftarrow $ EstimateVol($\polytope$), EstimateVol($\preimage_{\indomain}(\outset)$) \label{algline:initial_volume}\;
Dom $\leftarrow \{(\indomain, \polytope, \widehat{\volume}_{\preimage_{\indomain}(\outset)} - \widehat{\volume}_{\polytope})\}$ \tcp*{Priority queue}\label{algline:init_queue} 
\tcp{$\polytopeset_{\textnormal{Dom}}$ is the union of polytopes in Dom}
\While{$\textnormal{EstimateVol}(\polytopeset_{\textnormal{Dom}}) < \threshold$ \textnormal{\textbf{and}} $\textnormal{Iterations}\leq \maxiterations$ \label{algline:while}}{ 
    $\indomain_{\text{sub}}, \polytope, \text{Priority}$ $\leftarrow$ Pop(Dom) \label{algline:refine_start}\tcp*{Subregion with highest priority}
    \If{SplitOnInput}
    {
    $id$ $\leftarrow$ SelectInputFeature($\text{Feature}_{I}$) \tcp*{$\text{Feature}_{I}$ is the set of input features/dimensions}\label{algline:input_selection}
  }
  \Else{
    $id \leftarrow$ SelectReLUNode($\text{Node}_{Z}$)\tcp*{$\text{Node}_{Z}$ is the set of unstable ReLU nodes}\label{algline:relu_selection}
  }
[$\indomain_{sub}^{l}$,$\indomain_{sub}^{u}$] 
 $\leftarrow$ SplitOnNode($\indomain_{sub}$, $id$)\label{algline:split_node_id}\tcp*{Split on the selected node}
  [$\polytope^{l}, \polytope^{u}$] $\leftarrow$ GenUnderApprox([$\indomain_{sub}^{l}$,$\indomain_{sub}^{u}$],  $\outset$) \tcp*{Generate preimage}\label{algline:subdomain_polytope}
     [$\widehat{\volume}_{\polytope^{l}}, \widehat{\volume}_{\polytope^{u}}$] $\leftarrow$ EstimateVol([$\polytope^{l}, \polytope^{u}$])\label{algline:polytope_volume}\;
$\widehat{\volume}_{\preimage_{\indomain_{sub}^{l}}(\outset)}, \widehat{\volume}_{\preimage_{\indomain_{sub}^{u}}(\outset)} \leftarrow \textnormal{EstimateVol}(\preimage_{\indomain_{sub}^{l}}(\outset)), \textnormal{EstimateVol}(\preimage_{\indomain_{sub}^{u}}(\outset))$ \label{algline:preimage_volume}\;
    Dom $\leftarrow$ Dom $\cup$ \{($\indomain_{sub}^{l}, \polytope^{l}$,$\widehat{\volume}_{\preimage_{\indomain_{sub}^{l}}(\outset)} - \widehat{\volume}_{\polytope^{l}}$)\} $\cup$ \{($\indomain_{sub}^{u}, \polytope^{u}$,$\widehat{\volume}_{\preimage_{\indomain_{sub}^{u}}(\outset)} - \widehat{\volume}_{\polytope^{u}}$)\}\label{algline:refine_end}\tcp*{Disjoint polytope} 
    } \label{algline:update_approximation}
\KwRet{$\polytopeset_{\textnormal{Dom}}$}
\end{algorithm}

\subsection{Polytope Under-Approximation via Linear Relaxation}
\label{sec:poly_gen}

\begin{algorithm}[tb]
\caption{GenUnderApprox}\label{alg:genunderapprox}
\KwIn{List of subregions $\indomain$, Output set $\outset$, number of samples $\numsamples$}
\KwOut{List of polytopes $\textbf{\polytope}$}

$\textbf{\polytope} = []$\;
\For{\textnormal{subregion} $\indomain_{sub} \in \indomain$ \tcp{Parallel over subregions}}  {
    $[\underline{\spec_1}(\inpoint, \nnslope_1), ..., \underline{\spec_{\specnum}}(\inpoint, \nnslope_{\specnum})] \leftarrow$ LinearLowerBound($\indomain_{sub}, \outset$)\; \label{algline:linearbound}
    $\inpoint_1, ..., \inpoint_{\numsamples} \leftarrow $ Sample($\indomain_{sub}, \numsamples$)\;
   Loss$(\nnslope_1, ..., \nnslope_{\specnum}) \leftarrow -\sum_{j = 1, ..., \numsamples} \sigma(-\textnormal{LSE}(-\underline{\spec_1}(\inpoint_j, \nnslope_1), ..., -\underline{\spec_{\specnum}}(\inpoint_j, \nnslope_{\specnum}))$\label{algline:loss}\;
    $\nnslope_1^*, ..., \nnslope_{\specnum}^* \leftarrow \textnormal{Optimize}(\textnormal{Loss}(\nnslope_1, ..., \nnslope_{\specnum}))$\;
    $\textbf{\polytope} = \textnormal{Append}(\textbf{\polytope}, [\underline{\spec_1}(\inpoint, \nnslope_1^*) \geq 0, ..., \underline{\spec_{\specnum}}(\inpoint, \nnslope_{\specnum}^*) \geq 0, \inpoint \in \indomain_{sub}])$ \label{algline:append}
}

\KwRet{$\textbf{\polytope}$}
\end{algorithm}
We first show how to adapt linear relaxation techniques to efficiently generate valid under-approximations to the restricted preimage for a given input region $\indomain$. Recall that LiRPA methods enable us to obtain linear lower and upper bounds on the output of a neural network $\nn$, that is, $\lowerweight \inpoint + \lowerbias \leq \net(\inpoint) \leq \upperweight \inpoint + \upperbias$, where the linear coefficients depend on the input region $\indomain$.

Now, suppose that we are interested in computing an under-approximation to the restricted preimage, given the input hyperrectangle $\indomain = \{\inpoint \in \mathbb{R}^{\indim} | \inpoint \models \bigwedge_{i = 1}^{\indim} \boxcon_i \}$, and the output polytope specified using the half-space constraints $\lincon_i(\outpoint) = (\linconw_i^{T} \outpoint + \linconb_i \geq 0)$ for $ i = 1, ..., \specnum$ over the output space. Given a constraint $\lincon_i$, we append an additional linear layer at the end of the network $\nn$, which maps $\outpoint \mapsto \linconw_i^{T} \outpoint + \linconb_i$, such that the function $\spec_i: \mathbb{R}^{\indim} \to \mathbb{R}$ represented by the new network is $\spec_i(\inpoint) = \linconw_i^{T} \nn(\inpoint) + \linconb_i$. Then, applying LiRPA bounding to each $\spec_i$, we obtain lower bounds $\underline{\spec_i}(\inpoint) = \lowerweightsingle_i^T \inpoint + \lowerbiassingle_i$ for each $i$, such that $\underline{\spec_i}(\inpoint) \geq 0 \implies \spec_i(\inpoint) \geq 0$ for $\inpoint \in \indomain$. 
Notice that, for each $i = 1,..., \specnum$, $\lowerweightsingle_i^T \inpoint + \lowerbiassingle_i \geq 0$ is a half-space constraint in the input space. We conjoin these constraints, along with the restriction to the input region $\indomain$, to obtain a polytope 
    $\polytope_{\indomain}(\outset) := \{\inpoint| \bigwedge_{i =1}^{\specnum} (\underline{\spec_i}(\inpoint)  \geq 0 )\wedge \bigwedge_{i = 1}^{\indim} \boxcon_i(\inpoint) \}$.

\begin{restatable}{proposition}{propUnder}
$\polytope_{\indomain}(\outset)$ is an under-approximation to the restricted preimage $\preimage_{\indomain}(\outset)$.
\end{restatable}

\begin{example}\label{eg:underapprox}
Returning to Example \ref{eg:formulation}, the output constraints (for $i = 2, 3, 4$) are given by $\lincon_i = (\outpoint_1 - \outpoint_i \geq 0) = (\linconw_i^{T} \outpoint + \linconb_i \geq 0)$,  where $\linconw_i := e_1 - e_i$ (where $e_i$ is the $i^{\text{th}}$ standard basis vector) and $\linconb_i := 0$. Applying LiRPA bounding, we obtain the linear lower bounds $\underline{\spec_2}(\inpoint) = -3.79 \inpoint_1 + \inpoint_2 + 2.65 \geq 0; \underline{\spec_3}(\inpoint) = 0.34 \inpoint_1 - \inpoint_2 -0.60 \geq 0; \underline{\spec_4}(\inpoint) = -1.11 \inpoint_1 - \inpoint_2 + 1.99 \geq 0$ for each constraint. The intersection of these constraints, shown in Figure \ref{fig:before_optim}, represents the region where any input is guaranteed to satisfy the output constraints.
\end{example}

We generate the linear bounds in parallel over the output polyhedron constraints $i = 1, ..., \specnum$ using the \textit{backward mode} LiRPA \cite{zhang2018crown}, and store the resulting input polytope $\polytope_{\indomain}(\outset)$ as a list of constraints.  
This highly efficient procedure is used as a sub-routine \texttt{LinearLowerBound}  when generating a preimage under-approximation as a polytope union using Algorithm \ref{alg:genunderapprox}
(Line \ref{algline:linearbound}).

\subsection{Local Optimization} 
\label{sec:local_opt}

One of the key components behind the effectiveness of LiRPA-based bounds is the ability to efficiently improve the 
tightness of the bounding function by optimizing the relaxation parameters $\nnslope$, via projected gradient descent. 
In the context of local robustness verification,  
the goal is to optimize the 
concrete lower or upper bounds 
over the (sub)region $\indomain$ \cite{xu2020automated}, i.e.,
$\min_{\inpoint \in \indomain} \lowerweight(\nnslope) \inpoint + \lowerbias(\nnslope)$, where we explicitly note
the dependence of the linear coefficients on $\nnslope$.
In our case, we are instead interested in optimizing $\nnslope$ to refine
the polytope under-approximation, that is, increase its volume.
Unfortunately, computing the volume of a polytope exactly is a computationally expensive task, and requires specialized tools \cite{Augustin22polyvolume} that do not permit easy optimization with respect to the $\nnslope$ parameters.

To address this challenge, we propose to use 
statistical estimation.  
In particular, we sample $\numsamples$ points $\inpoint_1, ..., \inpoint_\numsamples$ uniformly from the input domain $\indomain$
then employ Monte Carlo estimation for the volume of the polytope approximation: 
\begin{align}\label{eq:stat_est}
&\widehat{\volume}(\polytope_{\indomain, \nnslope}(\outset)) = \frac{\sum_{i = 1}^{\numsamples}\mathds{1}_{\inpoint_i \in \polytope_{\indomain, \nnslope}(\outset)}}{\numsamples} \times \text{vol}(\indomain)
\end{align}
where we highlight the dependence of  $\polytope_{\indomain}(\outset) = \{\inpoint| \bigwedge_{i =1}^{\specnum} \underline{\spec_i}(\inpoint, \nnslope_i)  \geq 0 \wedge \bigwedge_{i = 1}^{\indim} \boxcon_i(\inpoint) \}$ on $\nnslope = (\nnslope_1, ..., \nnslope_\specnum)$,
and
$\nnslope_i$ are the $\nnslopesingle$-parameters for the linear relaxation of the neural network $\spec_i$ corresponding to the $i^{\textnormal{th}}$ half-space constraint in $\outset$.
However,
this is still non-differentiable w.r.t. $\nnslope$ due to the identity function. We now show how to derive a differentiable relaxation which is amenable to gradient-based optimization: 
\begin{align*}
    \widehat{\volume}(\polytope_{\indomain, \nnslope}(\outset)) &= \frac{\text{vol}(\indomain)}{\numsamples} \sum_{j = 1}^{\numsamples}\mathds{1}_{\inpoint_j \in \polytope_{\indomain, \nnslope}(\outset)} 
    = \frac{\text{vol}(\indomain)}{\numsamples} \sum_{j = 1}^{\numsamples} \mathds{1}_{\min_{i = 1, ... \specnum} \underline{\spec_i}(\inpoint_j, \nnslope_i) \geq 0}\\
    & \approx \frac{\text{vol}(\indomain)}{\numsamples} \sum_{j = 1}^{\numsamples} \sigma\left(\min_{i = 1, ... \specnum} \underline{\spec_i}(\inpoint_j, \nnslope_i)\right) \\
    & \approx \frac{\text{vol}(\indomain)}{\numsamples} \sum_{j = 1}^{\numsamples} \sigma\left(-\textnormal{LSE}( -\underline{\spec_1}(\inpoint_j, \nnslope_1), ..., -\underline{\spec_\specnum}(\inpoint_j, \nnslope_\specnum))\right)
\end{align*}
The second equality follows from the definition of the polytope $\polytope_{\indomain, \nnslope}(\outset)$; namely that a point is in the polytope if it satisfies $\underline{\spec_i}(\inpoint_j, \nnslope_i) \geq 0$ for all $i = 1, ..., \specnum$, or equivalently, $\min_{i = 1, ... \specnum} \underline{\spec_i}(\inpoint_j, \nnslope_i) \geq 0$. After this, we approximate the identity function using a sigmoid relaxation, where $\sigma(y) := \frac{1}{1 + e^{-y}}$, as is commonly done in machine learning to define classification losses. Finally, we approximate the minimum over specifications using the log-sum-exp (LSE) function. The log-sum-exp function is defined by $LSE(y_1, ..., y_{\specnum}) := \log(\sum_{i = 1, ..., \specnum} e^{y_i})$, and is a differentiable approximation to the maximum function; we employ it to approximate the minimization by adding the appropriate sign changes. The final expression is now a differentiable function of $\nnslope$. We employ this as the loss function in Algorithm \ref{alg:genunderapprox} (Line \ref{algline:loss}) for generating a polytope approximation, and optimize volume using projected gradient descent.

\begin{example}
\label{eg:optim}
    We revisit the vehicle parking problem in Example~\ref{eg:formulation}.
    Figure \ref{fig:before_optim} and \ref{fig:after_optim} show the computed under-approximations before and after local optimization. 
    We can see that the bounding planes for all three specifications are optimized, which 
    effectively improves the approximation quality. 
\end{example}

\subsection{Global Branching and Refinement} \label{sec:branching}

As LiRPA performs crude linear relaxation,
the resulting bounds can be quite loose even with $\nnslope$-optimization, meaning that the 
polytope approximation $\polytope_{\indomain}(\outset)$ is unlikely to constitute a tight under-approximation to the preimage.
To address this challenge, we employ a divide-and-conquer approach that
iteratively refines our under-approximation of the preimage. Starting from the initial region $\indomain$ represented at the root, our method generates a tree by iteratively partitioning a subregion $\indomain_{sub}$ represented at a leaf node into two smaller subregions $\indomain_{sub}^{l}, \indomain_{sub}^{u}$, which are then attached as children to that leaf node. In this way, the subregions represented by all leaves of the tree are disjoint, such that their union is the initial region $\indomain$.

For each leaf subregion $\indomain_{sub}$ we compute, using  LiRPA bounds (Line \ref{algline:linearbound}, Algorithm \ref{alg:genunderapprox}), an associated polytope that under-approximates the preimage in $\indomain_{sub}$. Thus, irrespective of the number of refinements performed, the union of the polytopes corresponding to all leaves forms an \emph{anytime} DUP under-approximation $\polytopeset$ to the preimage in the original region $\indomain$.
The process of refining the subregions  continues until an appropriate termination criterion is met.

Unfortunately, even with a moderate number of input dimensions or unstable ReLU nodes, na\"ively splitting along all input- or ReLU-planes quickly becomes computationally infeasible. For example, splitting a $\indim$-dimensional hyperrectangle using
bisections along each dimension results in $2^d$ subdomains to approximate. It 
thus becomes crucial to identify the subregion splits that have the most impact on the quality of the under-approximation.
Another important aspect is how to prioritize which leaf subregion to split. We describe these in turn.

\textbf{Subregion Selection.}
Searching through all leaf subregions at each iteration is computationally too expensive. Thus, we propose a subregion selection strategy that prioritizes splitting subregions according to (an estimate of) the difference in volume
between the exact preimage $\preimage_{\indomain_{sub}}(\outset)$ and the (already computed) polytope approximation $\polytope_{\indomain_{sub}}(\outset)$ on that subdomain, that is:
\begin{equation} \label{eqn:priority}
\textnormal{Priority}(\indomain_{sub})  = \volume(\preimage_{\indomain_{sub}}(\outset)) - \volume(\polytope_{\indomain_{sub}}(\outset))
\end{equation}
which measures the 
gap between the polytope under-approximation and the optimal approximation, namely, the preimage itself.

Suppose that a particular leaf subdomain attains the maximum of this metric among all leaves, and we partition it into two subregions $\indomain_{sub}^l, \indomain_{sub}^u$, which we approximate with polytopes $\polytope_{\indomain_{sub}^l}(\outset), \polytope_{\indomain_{sub}^u}(\outset)$. 
As tighter intermediate concrete bounds, and thus linear bounding functions, can be computed on the partitioned subregions,
the polytope approximation on each subregion will be 
refined compared with the single polytope restricted to that subregion.

\begin{restatable}{proposition}{propPriority} \textbf{\label{prop:volume_guarantee}}
    Given any subregion $\indomain_{sub}$ with polytope approximation $\polytope_{\indomain_{sub}}(\outset)$, and its children $\indomain_{sub}^l, \indomain_{sub}^u$ with polytope approximations $\polytope_{\indomain_{sub}^l}(\outset), \polytope_{\indomain_{sub}^u}(\outset)$ respectively, it holds that:
    \begin{equation} \label{eqn:no_fragmentation}
        \polytope_{\indomain_{sub}^l}(\outset) \cup \polytope_{\indomain_{sub}^u}(\outset) \supseteq \polytope_{\indomain_{sub}}(\outset)
    \end{equation}
\end{restatable}

\begin{restatable}{corollary}{corVolume}
    In each refinement iteration,
    the volume of the polytope approximation $\polytopeset_{Dom}$ does not decrease.
\end{restatable}

Since computing the volumes in Equation \ref{eqn:priority} is intractable, we sample $\numsamples$ points $\inpoint_1, ..., \inpoint_\numsamples$ uniformly from the subdomain $\indomain_{sub}$ and 
employ Monte Carlo estimation to estimate the volume for both the preimage and the polytope approximation using the same set of samples,
i.e., 
$
\widehat{\volume}(\preimage_{\indomain_{sub}}(\outset)) = \text{vol}(\indomain_{sub}) \times \frac{1}{\numsamples} \sum_{i = 1}^{\numsamples}\mathds{1}_{\inpoint_i \in \preimage_{\indomain_{sub}}(\outset)}$, and $
\widehat{\volume}(\polytope_{\indomain_{sub}}(\outset)) = \text{vol}(\indomain_{sub}) \times  \frac{1}{\numsamples}\sum_{i = 1}^{\numsamples}\mathds{1}_{\inpoint_i \in \polytope_{\indomain_{sub}}(\outset)} 
$. 
We stress that volume estimation is only used to prioritize subregion selection, and does not affect the soundness of our method.

\textbf{Input Splitting.}
Given a subregion (hyperrectangle) defined by lower and upper bounds $\inpoint_i \in [\featurelowerbound_i, \featureupperbound_i]$ for all dimensions $i = 1, ..., \indim$, input splitting partitions it
into two subregions by 
cutting along some feature $i$. 
This splitting procedure will produce two subregions
which are similar to the original subregion, but have
updated bounds $[\featurelowerbound_i, \frac{\featurelowerbound_i + \featureupperbound_i}{2}], [\frac{\featurelowerbound_i + \featureupperbound_i}{2}, \featureupperbound_i]$ for feature $i$ instead. 
In order to 
determine which feature/dimension to split on, we propose a greedy  
strategy.
Specifically, for each feature, we generate a pair of polytopes for the two subregions resulting from the split, and choose the feature that results in the greatest total volume of the polytope pair. 
In practice,
another commonly-adopted splitting heuristic  is to select the dimension with the longest edge~\cite{Bunel20BaB}, that is, to select feature $i$ with the largest range: $\arg \max_i (\featureupperbound_i-\featurelowerbound_i)$. However, this method falls short in per-iteration approximation volume improvement compared to our greedy strategy.

\begin{figure}[t]
     \centering
     \begin{subfigure}[b]{0.23\textwidth}
         \centering
         \includegraphics[width=\textwidth]{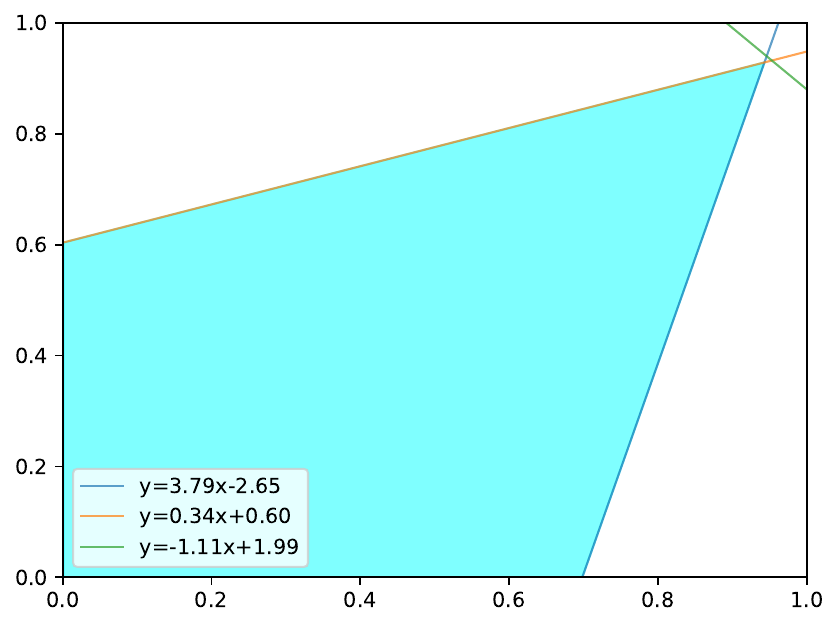}
         \caption{No optimization}
         \label{fig:before_optim}
     \end{subfigure}
     \begin{subfigure}[b]{0.23\textwidth}
         \centering
         \includegraphics[width=\textwidth]{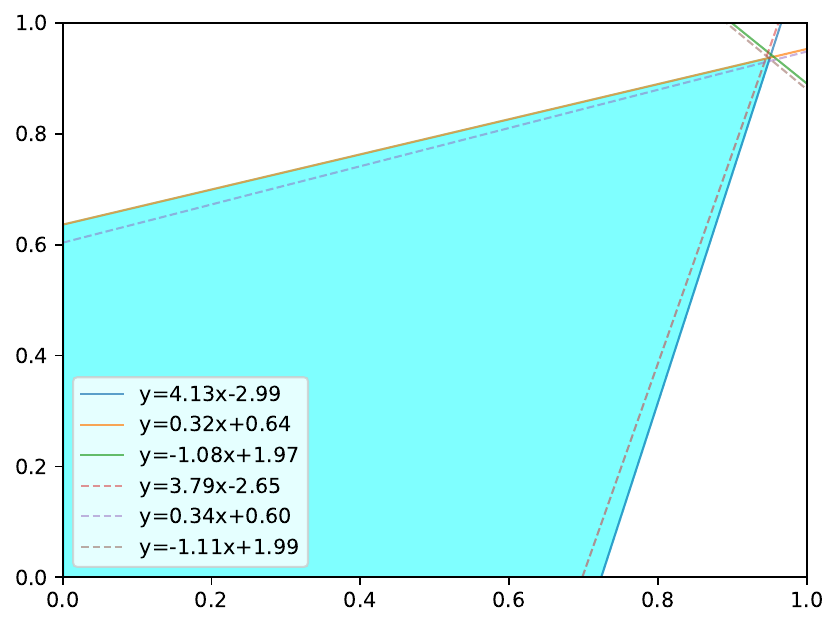}
         \caption{After optimization}
         \label{fig:after_optim}
     \end{subfigure}
     \begin{subfigure}[b]{0.23\textwidth}
         \centering
         \includegraphics[width=\textwidth]{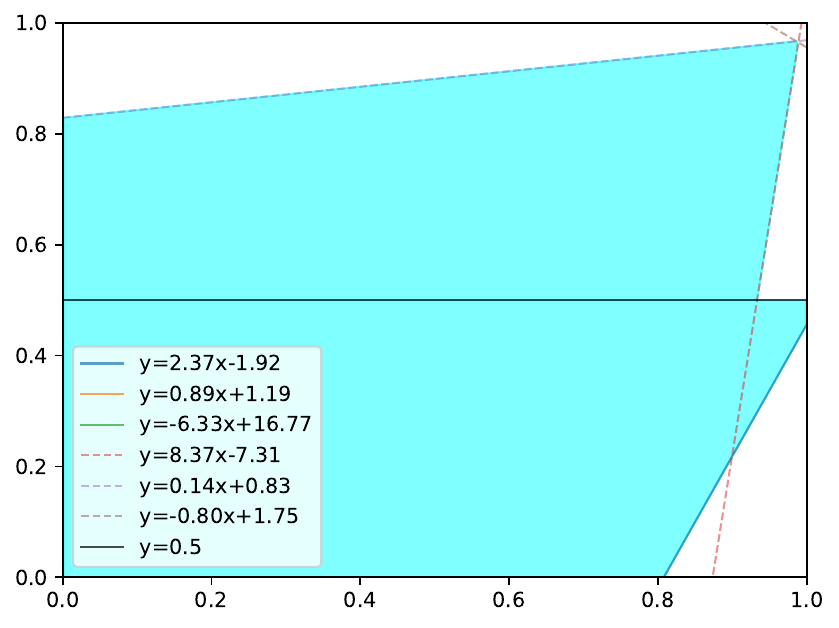}
         \caption{Input split}
         \label{fig:after_branch}
     \end{subfigure}
     \begin{subfigure}[b]{0.23\textwidth}
         \centering
         \includegraphics[width=\textwidth]{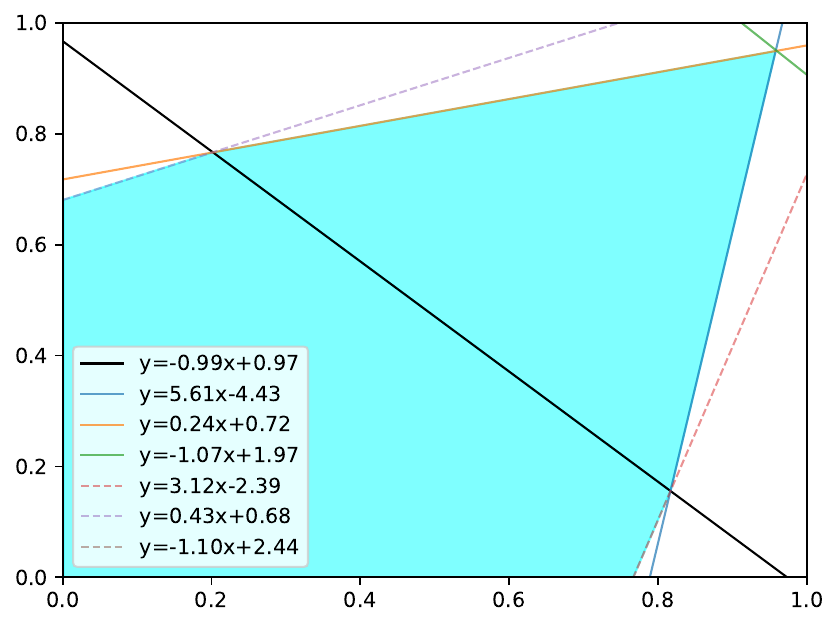}
         \caption{ReLU split}
         \label{fig:relu_split}
     \end{subfigure}
        \caption{Refinement and optimization for preimage approximation.}
        \label{fig:global_branching}
\end{figure}
\begin{example}
\label{eg:branching}
We revisit the vehicle parking problem in Example~\ref{eg:formulation}.
Figure \ref{fig:after_optim} shows the polytope under-approximation computed on the input region $\indomain$ before refinement, where each solid line represents the bounding plane for each output specification ($\outpoint_1 - \outpoint_i \geq 0 $). 
Figure \ref{fig:after_branch} depicts the refined approximation by splitting the input region along the vertical axis, where the solid and dashed lines represent the 
bounding planes for the two resulting subregions. It can be seen that the total volume of the under-approximation has improved significantly.
\end{example}

\textbf{Intermediate ReLU Splitting.}
Refinement through splitting on input features is adequate for low-dimensional input problems such as reinforcement learning agents. However, it may be infeasible to generate sufficiently fine subregions for high-dimensional domains. 
We thus propose an algorithm for ReLU neural networks that uses intermediate ReLU splitting for preimage refinement.
After determining a subregion for refinement, we partition the subregion based upon the pre-activation value of an intermediate unstable neuron $\intervar=0$.
As a result, the original subregion $\indomain_{sub}$ is split into two new subregions
$\indomain^{+}_{\intervar}=\{\inpoint \in \indomain_{sub}~|~\intervar=\preact^{(i)}_j(\inpoint) \ge 0\}$
and $\indomain^{-}_{\intervar}=\{\inpoint \in \indomain_{sub}~|~\intervar=\preact^{(i)}_j(\inpoint) < 0\}$.\footnote{To obtain a polytope under-approximation, we can utilize linear lower/upper bounds on $\preact^{(i)}_j(\inpoint)$ as an approximation to the subregion boundary.}

In this procedure, the order of splitting unstable ReLU neurons can greatly influence the refinement quality and efficiency. 
Existing heuristic methods of ReLU prioritization select ReLU nodes that lead to greater improvement in the final bound (maximum or minimum value) of the neuron network on the input domain~\cite{Bunel20BaB}, i.e. $\min_{\inpoint \in \indomain} \nnunder(x)$.
However, these ReLU prioritization methods are not effective for preimage analysis, because 
our objective is instead to refine the overall preimage approximation. 
We thus propose a heuristic method to prioritize unstable ReLU nodes for preimage refinement.
Specifically, we compute (an estimate of) the volume difference 
between the split subregions
$\lvert \volume(\indomain^{+}_{\intervar})-\volume(\indomain^{-}_{\intervar})\rvert$, using a single forward pass for a set of sampled datapoints from the input domain; note that this is bounded above by the total subregion volume $\volume(\indomain_{sub})$.
We then propose to select the ReLU node that minimizes this difference.
Intuitively, this choice results in balanced subdomains after splitting. 

Another advantage of ReLU splitting is that we can replace the unstable neuron bound $\underline{c} \preact^{(i)}_j(\inpoint) + \underline{d} \leq \postact^{(i)}_j(\inpoint) \leq \overline{c} \preact^{(i)}_j(\inpoint) + \overline{d}$ with the exact linear function $\postact^{(i)}_j(\inpoint) = \preact^{(i)}_j(\inpoint)$ and $\postact^{(i)}_j(\inpoint) = 0$, respectively, as shown in Figure \ref{fig:linear_relaxation} (unstable to stable). This can then tighten the linear bounds for the other neurons, thus tightening the under-approximation on each subdomain.

\begin{example}
\label{eg:reluSplit}
We now apply our algorithm with ReLU splitting to the vehicle parking problem in Example~\ref{eg:formulation}.
Figure \ref{fig:relu_split} shows the refined preimage polytope by adding the splitting plane (black solid line) along the direction of a selected unstable ReLU node. 
Compared with Figure \ref{fig:after_optim}, we can see that the volume of the approximation is improved. 

\end{example}

\remark[Preimage Over-approximation] While Algorithms \ref{alg:main} and \ref{alg:genunderapprox} 
focus on preimage under-approximations, they can be easily configured to generate over-approximations with two key modifications.
Firstly, we generate polytope over-approximations by using LiRPA to propagate a linear \textit{upper} bound $\overline{\spec_i}(\inpoint) = \upperweightsingle_i^T \inpoint + \upperbiassingle_i$ for each output constraint, such that $\spec_i(\inpoint) \geq 0 \implies \overline{\spec_i}(\inpoint) \geq 0 $ for $\inpoint \in \indomain$. Secondly, the 
refinement and
optimization objective is to \emph{minimize} the volume of the over-approximation instead of maximizing the volume as in the case of under-approximation.

\subsection{Overall Algorithm}\label{sec:overall}

Our overall preimage approximation method is summarized in Algorithm \ref{alg:main}. It takes as input a neural network $f$, input region $\indomain$, output region $\outset$, target polytope volume threshold $v$ (a proxy for approximation precision), termination iteration number $\maxiterations$, and a Boolean indicating whether to use input or ReLU splitting, and returns a disjoint polytope union $\polytopeset$ representing an underapproximation to the preimage.

The algorithm initiates and maintains a priority queue of (sub)regions according to Equation~\ref{eqn:priority}.
The \textit{initialization} step (Lines \ref{algline:initial_polytope}-\ref{algline:init_queue}) generates an initial polytope approximation on the whole region using Algorithm \ref{alg:genunderapprox} (Sections \ref{sec:poly_gen}, \ref{sec:local_opt}). Then, the \textit{preimage refinement} loop (Lines \ref{algline:while}-\ref{algline:refine_end}) partitions a subregion in each iteration, with the preimage restricted to the child subregions then being re-approximated (Line \ref{algline:split_node_id}-\ref{algline:subdomain_polytope}). In each iteration, we choose the region to split (Line \ref{algline:refine_start}) and the splitting plane to cut on (Line \ref{algline:input_selection} for input split and Line \ref{algline:relu_selection} for ReLU split), as
detailed in Section \ref{sec:branching}. The preimage under-approximation is then updated by computing the priorities for each subregion (by approximating volumes) (Lines \ref{algline:polytope_volume}-\ref{algline:update_approximation}). The loop terminates and the approximation returned when the target volume threshold $v$ or maximum iteration limit $R$ is reached.

\subsection{Quantitative Verification} \label{sec:verif}

We now show how to use our efficient preimage under-approximation method (Algorithm \ref{alg:main}) to verify a given quantitative property $(\inset, \outset, \proportion)$, where $\outset$ is a polyhedron, $\inset$ a polytope and $p$ the desired proportion value, summarized in Algorithm \ref{alg:verify}. To simplify assume that $\inset$ is a hyperrectangle, so that we can take $\indomain = \inset$ (the case of general polytopes is discussed in the Appendix). 
\begin{algorithm}[tb]
\caption{Quantitative Verification}\label{alg:verify}
\KwIn{Neural network $f$, Property $(\inset, \outset, \proportion)$, Maximum iterations $\maxiterations$}
\KwOut{Verification Result $\in \{\textnormal{True, False, Unknown}\}$}

$\volume(\inset) \gets \textnormal{ExactVolume}(\inset)$\;
$\indomain \gets \textnormal{OuterBox}(\inset)$ \label{algline:outerbox}\tcp*{For general polytopes $\inset$}
$\polytopeset \gets \textnormal{InitialRun}(f, \indomain, \outset)$\;
\While{$\textnormal{Iterations} \leq \maxiterations$} {
    $\polytopeset \gets \textnormal{Refine}(f, \polytopeset, \indomain,\outset)$\;
    \If{$\textnormal{EstimateVolume}(\polytopeset) \geq \proportion \times \volume(\inset)$}{
        \If{$\textnormal{ExactVolume}(\polytopeset) \geq \proportion \times \volume(\inset)$} {
            \KwRet{\textnormal{True}}
        }
    }
    \If{\textnormal{AllReLUSplit}}{
    \KwRet{\textnormal{False}}
    }
}
\KwRet{\textnormal{Unknown}}
\end{algorithm}
We utilize Algorithm \ref{alg:main} by setting the volume threshold to $\proportion \times \volume(\inset)$, such that we have $\frac{\widehat{\volume}(\polytopeset)}{\volume(\inset)} \geq \proportion$ if the algorithm terminates before reaching the maximum number of iterations. 
However, 
the Monte Carlo estimates of volume 
cannot provide a sound guarantee that $\frac{\volume(\polytopeset)}{\volume(\inset)} \geq \proportion$. To resolve this problem, we propose to run exact volume computation \cite{Barber96Qhull} only when the Monte Carlo estimate reaches the threshold. If the exact volume $\volume(\polytopeset) \geq \proportion \times \volume(\inset)$, then the property is verified. Otherwise, we continue running the preimage refinement. 

In Algorithm \ref{alg:verify}, \texttt{InitialRun} generates an initial approximation to the preimage as in Lines \ref{algline:initial_polytope}-\ref{algline:init_queue} of Algorithm 1, and \texttt{Refine} performs one iteration of approximation refinement (Lines \ref{algline:refine_start}-\ref{algline:refine_end}). Termination occurs when we have verified or falsified the quantitative property, or when the maximum number of iterations has been exceeded.

\begin{restatable}{proposition}{propSound}
    Algorithm \ref{alg:verify} is sound for quantitative verification with input splitting.
\end{restatable}

\begin{restatable}{proposition}{propComplete}
    Algorithm \ref{alg:verify} is sound and complete for quantitative verification on piecewise linear neural networks with ReLU splitting.
\end{restatable}

\section{Experiments}
\label{sec:evaluation}
We have implemented our approach as a prototype tool \footnote{The source code
is at \texttt{https://github.com/Zhang-Xiyue/PreimageApproxForNNs.}} for preimage approximation for 
polyhedron-type 
output sets/specifications. 
In this section, 
we perform experimental evaluation of the proposed approach on a set of benchmark tasks and demonstrate its 
effectiveness in approximation generation 
and its application to
quantitative analysis of neural networks.

\subsection{Benchmark and Evaluation Metric}\label{sec:eval_metric}
We evaluate our preimage analysis approach on a benchmark of reinforcement learning and image classification tasks. 
Besides the vehicle parking task \cite{Ayala11vehicle} shown in the running example,
we use the following (trained) benchmarks:
(1) aircraft collision avoidance system (VCAS) \cite{Julian19nncontrol} with 9 feed-forward neural networks (FNNs); (2) neural network controllers from VNN-COMP 2022 \cite{vnn22} for three reinforcement learning tasks (Cartpole, Lunarlander, and Dubinsrejoin)~\cite{Brockman16}; and (3) the neural network from VNN-COMP 2022 for MNIST classification.
Details of the models and additional experiments can be found in the Appendix.

\textbf{Evaluation Metric}
To evaluate the quality of the preimage approximation, 
we define the \textit{coverage ratio} to be the ratio of volume covered to the volume of the exact preimage,
i.e., 
    $\cov(\polytopeset, \preimage_{\indomain}(\outset)) := \frac{\volume(\polytopeset)}{\volume(\preimage_{\indomain}(\outset))}$.
Note that this is a normalized measure for assessing the quality of the approximation, 
as shown in Algorithm~\ref{alg:verify} when comparing with target coverage proportion $\proportion$ for termination of the refinement loop, and not as a measure for formal verification guarantees.
In practice, we estimate $\volume(\preimage_{\indomain}(\outset))$ 
as $\widehat{\volume}(\preimage_{\indomain}(\outset)) = \text{vol}(\indomain)\times \frac{1}{\numsamples} \sum_{i=1}^{\numsamples} \mathds{1}_{\nn(\inpoint_i) \in \outset}$, where $\inpoint_1, ... \inpoint_{\numsamples}$ are samples from $\indomain$.
In Algorithm \ref{alg:main}, 
the target volume  (stopping criterion) is set as $\threshold = \targetcov \times \widehat{\volume}(\preimage_{\indomain}(\outset)$, where $\targetcov$ is the \textit{target coverage ratio}.

\subsection{Evaluation}\label{sec:eval}
\subsubsection{Effectiveness in Preimage Approximation with Input Split}

We apply Algorithm~\ref{alg:main}
with input splitting to 
the input bounding problem for
low-dimensional reinforcement learning tasks to evaluate its effectiveness.
For comparison, we also run the exact preimage (Exact) \cite{Matoba20Exact} and preimage over-approximation (Invprop) \cite{ProveBound,invprop} methods.

\begin{table}[tb]
\caption{Performance comparison in preimage generation.}\label{tab:compare_baseline}
\centering
\begin{tabular}{c|c c|c c|c c c} 
\toprule
\multirow{2}{*}{\textbf{Models}}  
        & \multicolumn{2}{c|}{\textbf{Exact}}
        & \multicolumn{2}{c|}{\textbf{Invprop}}
         & \multicolumn{3}{c}{\textbf{Our}} \\
\cmidrule{2-8}
         & \textbf{\#Poly} & \textbf{Time} 
        & \textbf{Time} & \textbf{Cov(\%)} 
        & \textbf{\#Poly} & \textbf{Time} & \textbf{Cov(\%)} \\ 
\midrule
Vehicle (FNN $1\times 20$)& 10 & 3110.979
 & 2.642 & 92.1
& 4  & 1.175 & 95.7 \\ 
VCAS (FNN $1\times 21$)
& 131 & 6363.272 
 & - & -
& 12 & 11.281 & 91.0\\
\bottomrule
\end{tabular}
\end{table}
\textit{Vehicle Parking \& VCAS.}
Table~\ref{tab:compare_baseline} presents experimental results on the vehicle parking and VCAS tasks.
In the table, we show the number of polytopes (\#Poly) in the preimage, computation time (Time(s)), and
the approximate coverage ratio (Cov(\%)) when the preimage approximation algorithm terminates with target coverage 90\%.
Compared with the exact method, 
our approach yields \textit{orders-of-magnitude} improvement in efficiency. It can also characterize the preimage with much fewer (and also disjoint) polytopes (average reduction of 91.1\% for VCAS).

The Invprop method \cite{ProveBound} cannot be directly applied as it computes preimage over-approximations.
We adapt it to produce an under-approximation by computing over-approximations for the complement of each output constraint; the resulting approximation is then the complement of a
union of polytopes,
rather than a DUP.
On the 2D vehicle parking task, we find that the results (see Table~\ref{tab:compare_baseline}) are comparable with ours in time and approximation coverage.
Their implementation currently only supports two-dimensional input tasks~\cite{invprop}.  
While their algorithm, which employs input splitting, can in theory be extended to higher-dimensional tasks, a significant unaddressed technical challenge is in how to choose the input splits effectively in high dimensions. This is confounded by the fact that, to generate an under-approximation, we need separate runs of their algorithm for each output constraint. In contrast, 
our method naturally incorporates a principled splitting and refinement strategy, and can also effectively employ ReLU splitting for further  scalability, as we will show below.
Our method can also be configured to generate over-approximations (Section~\ref{sec:branching}, Remark 1).

\begin{table}[t]
\caption{Performance of preimage approximation for reinforcement learning tasks.}\label{tab:rl-tasks}
\centering
\begin{tabular}{cc|c|c|c|c} 
\toprule
\textbf{Task} & \textbf{Property} & \textbf{Config} & \textbf{\#Poly} & \textbf{Cov(\%)} & \textbf{Time} \\ 
\midrule
\multirow{3}{*}{\makecell{Cartpole\\ (FNN $2 \times 64$)}} & \multirow{3}{*}{$\{y\in \mathbb{R}^2|~ y_1 \ge y_2\}$} & $\dot{\theta} \in [-2,-1]$ & 8 & 82.0 & 8.933\\ 
& & $\dot{\theta} \in [-2,-0.5]$ & 17 & 75.5 & 14.527\\ 
& & $\dot{\theta} \in [-2,0]$ &32 & 76.5 & 27.344\\ 
\midrule
\multirow{3}{*}{\makecell{Lunarlander \\ (FNN $2 \times 64$)}}
 & 
\multirow{3}{*}{$\{y\in \mathbb{R}^4| \wedge_{i \in \{1,3,4\}}  y_2 \ge y_i\}$} & $\dot{v} \in [-0.5, 0]$ & 38 & 75.5 & 34.311\\ 
& & $\dot{v} \in [-1, 0]$ & 71 & 75.1 & 63.333\\ 
& & $\dot{v} \in [-2, 0]$ & 159 & 75.0 & 134.929 \\
\midrule
\multirow{3}{*}{\makecell{Dubinsrejoin\\ (FNN $2 \times 256$)}}
& 
\multirow{3}{*}{\makecell{$\{y\in \mathbb{R}^8|\wedge_{i \in [2,4]} ~ y_1 \ge y_i $ \\
$\quad \quad \quad \,\bigwedge \wedge_{i \in [6,8]} ~y_5 \ge y_i\}$}} & $x_v \in [-0.1,0.1]$ & 26 & 75.8 & 34.558\\ 
& & $x_v \in [-0.2,0.2]$ & 61 & 75.4 & 78.437\\
& & $x_v \in [-0.3,0.3]$ & 1002 & 57.6 & 1267.272\\
\bottomrule
\end{tabular}
\end{table}

\textit{Neural Network Controllers.}
In this experiment, we consider preimage under-approximation for neural network controllers in reinforcement learning tasks. 
Note that \cite{Matoba20Exact} (Exact) is unable to deal with neural networks of these sizes and \cite{ProveBound,invprop} (Invprop) does not support these higher-dimensional input domains.
Table~\ref{tab:rl-tasks} summarizes the experimental results.  
We evaluate 
Algorithm~\ref{alg:main} with input split
on a range of tasks/properties and configurations of the input region (e.g., angular velocity $\dot{\theta}$ for Cartpole). 
Empirically, for the same coverage ratio, our method requires a number of polytopes and time roughly linear in the input region size, with the exception of Dubinsrejoin, where the larger number of output constraints and larger network size contribute to greater relaxation error.

\noindent\textbf{MNIST Preimage Approximation with ReLU Split}
Next, we evaluate the scalability of Algorithm \ref{alg:main} with ReLU splitting
by applying it to MNIST image classifiers. To our knowledge, this is the first time preimage computation has been attempted for this challenging, high-dimensional task.

\begin{table}[t]
    \caption{Refinement with ReLU split for MNIST (FNN $6 \times 100$)}.
    \label{tab:image_task}
    \centering
    \begin{tabular}{c|c|c|c||c|c|c|c}
        \toprule
         $L_{\infty}$ \textbf{attack}
        &\textbf{\#Poly}
         & \textbf{Cov(\%)}
         & \textbf{Time} 
        &  \textbf{Patch attack}
        &\textbf{\#Poly}
         & \textbf{Cov(\%)}
         & \textbf{Time}\\
        \midrule
          0.05
        & 2
        & 100.0
        & 3.107
        & $3 \times 3$(center)
        & 1
        & 100.0
        & 2.611
        \\
         0.07
        & 247
        & 75.2
         & 121.661
           & $4 \times 4$(center)
        & 678
        & 38.2
         & 455.988
         \\
         0.08
        & 522
        & 75.1
        & 305.867
                & $6 \times 6$(corner)
        & 2
        & 100.0
        & 9.065
        \\
         0.09
         & 733
         & 16.5
         & 507.116
          & $7 \times 7$(corner)
         & 7
         & 84.2
         & 10.128
         \\
        \bottomrule
    \end{tabular}
\end{table}

Table \ref{tab:image_task} summarizes the evaluation results for
two types of image attacks: $l_{\infty}$ and patch attack.
For $L_{\infty}$ attacks, bounded perturbation noise is applied to all image pixels.
The patch attack applies only to a smaller patch area but allows arbitrary perturbations covering the whole valid range $[0,1]$. The task is then to produce a DUP under-approximation of the perturbation region that is guaranteed to be classified correctly.
For $L_{\infty}$ attack, our approach generates a preimage approximation that achieves the targeted coverage of $75\%$ for noise up to 0.08. Notice that, from e.g. $0.05$ to $0.07$, the volume of the input region increases by tens of orders of magnitude due to the high dimensionality. The fact that the number of polytopes and computation time remains manageable is  due to the effectiveness of ReLU splitting.
Interestingly, for the patch attack, we observe that the number of polytopes required increases sharply when increasing the patch size at the center of the image, while this is not the case for patches in the corners of the image. We hypothesize this is due to the greater influence of central pixels on the neural network output, and correspondingly a greater number of unstable neurons over the input perturbation space.

\begin{table}[t]
    \caption{Comparison with a robustness verifier.}
    \label{tab:compare_verifier}
    \centering
    \begin{tabular}{c|cc|ccc}
        \toprule
        \multirow{2}{*}{\textbf{Task}}
        &\multicolumn{2}{c|}{\textbf{$\alpha,\beta$-CROWN}}
         & \multicolumn{3}{c}{\textbf{Our}} \\
        \cmidrule{2-6}
        & Result
        & Time
        & Cov(\%)
        & \#Poly
        & Time\\
        \midrule
        Cartpole ($\dot{\theta} \in [-1.642,-1.546]$) 
         & yes
         & 3.349
        & 100.0
        & 1
         & 1.137
        \\        
         \midrule
        Cartpole ($\dot{\theta} \in [-1.642,0]$)
         &no
         & 6.927
        & 94.9
        & 2
         & 3.632
        \\
        \midrule
        MNIST ($L_{\infty}$ 0.026)
        & yes
         & 3.415
        & 100.0
        & 1
         & 2.649
         \\
         \midrule
        MNIST ($L_{\infty}$ 0.04)
        & unknown
         & 267.139
        & 100.0
        & 2
         & 3.019
         \\
        \bottomrule
    \end{tabular}
\end{table}
\noindent\textbf{Comparison with Robustness Verifiers}
We now illustrate empirically the utility of preimage computation in robustness analysis compared to robustness verifiers. Table \ref{tab:compare_verifier} shows comparison results with $\alpha,\beta$-CROWN,  winner of the VNN competition~\cite{vnn22}.
We set the tasks according to the problem instances from VNN-COMP 2022 for local robustness verification (localized perturbation regions).
For Cartpole, $\alpha,\beta$-CROWN can provide a verification guarantee (yes/no or safe/unsafe) for both of the problem instances. However, in the case where the robustness property does not hold, our 
method 
explicitly generates a preimage approximation in the form of a disjoint polytope union (where correct classification is guaranteed), and
covers $94.9\%$ of the exact preimage. 
For MNIST, while the smaller perturbation region is successfully verified, $\alpha,\beta$-CROWN with tightened intermediate bounds by MIP solvers returns unknown with a timeout of 300s for the larger region. 
In comparison,
our algorithm provides a concrete union of polytopes where the input is guaranteed to be correctly classified, which we find covers 100$\%$ of the input region (up to sampling error). Note also (Table \ref{tab:image_task}) that our algorithm can produce non-trivial under-approximations for input regions far larger than $\alpha, \beta$-CROWN can verify.

\noindent\textbf{Quantitative Verification}
We now demonstrate the application of our preimage generation framework to quantitative verification of the property $(\inset, \outset, \proportion)$; that is, to check whether $\nn(\inpoint) \in \outset$ for at least proportion $\proportion$ of input values $\inpoint \in \inset$. This leverages the disjointness of our approximation, such that we can exactly compute the volume covered by exactly computing the volume of each polytope.

\textit{Vehicle Parking.}
We 
consider the quantitative property with input set $\inset= \{x \in \mathbb{R}^{2}~|~x \in [0,1]^2\}$, output set  $\outset=\{\outpoint \in \mathbb{R}^{4}|\bigwedge_{i = 2}^{4} y_1 - \ y_i \geq 0\}$,
and quantitative proportion $p=0.95$.
We use Algorithm \ref{alg:verify} to verify this property, with iteration limit
 1000.
The computed under-approximation is a union of two polytopes, which takes 0.942s to reach the target coverage.
We then compute the exact volume ratio of the under-approximation against the input region.
The final quantitative proportion reached by our under-approximation is 95.2\%, verifying the quantitative property.

\textit{Aircraft Collision Avoidance.}
In this example, we consider the VCAS system
and a scenario where the two aircraft have negative relative altitude from intruder to ownship ($h \in [-8000, 0]$), the ownship aircraft has a positive climbing rate $\dot{h_A} \in[0,100]$ and the intruder has a stable negative climbing rate $\dot{h_B}=-30$, and time to the loss of horizontal separation is $t \in [0,40]$, which defines the input region $\inset$.
For this scenario, the correct advisory is ``Clear Of Conflict'' (COC).
We apply Algorithm \ref{alg:verify} to verify the quantitative property where $\outset=\{\outpoint \in \mathbb{R}^{9} | \bigwedge_{i = 2}^{9} y_1 - y_i \geq 0\}$ and the 
proportion $p=0.9$, with an iteration limit of 1000. 
The under-approximation computed is a union of 6 polytopes, which takes 5.620s to reach the target coverage.
The exact quantitative proportion reached by the generated under-approximation is 90.8\%, which verifies the quantitative property.

\section{Related Work}
Our paper is related to a series of works on 
robustness verification of neural networks.
To address the scalability issues with
\textit{complete} verifiers \cite{huang2017safety,katz2017reluplex,tjeng2019evaluating} based on constraint solving,
convex relaxation \cite{salman2019convex} has been used for developing highly efficient \textit{incomplete} verification methods \cite{zhang2018crown,wong2018provable,singh2019deeppoly,xu2020automated}.
Later works employed the branch-and-bound (BaB) framework \cite{bunel2018unified,Bunel20BaB} to achieve completeness, using incomplete methods for the bounding procedure \cite{xu2021fast,wang2021beta,Ferrari22MultiBaB}. In this work, we adapt convex relaxation for efficient preimage approximation. Further, our divide-and-conquer procedure is analogous to BaB, but focuses on maximizing covered volume rather than maximizing a function value.
There are also works
that have sought to define a weaker notion of local robustness known as \textit{statistical robustness} \cite{webb2018statistical,mangal2019probabilistic}, which requires that a proportion of points under some perturbation distribution around an input point are classified in the same way. 
Verification of statistical robustness is typically achieved by sampling and statistical guarantees \cite{webb2018statistical,baluta2021quantitative,tit2021corruption,yang2021quantitative}.
In this paper, we apply our symbolic approximation approach to quantitative analysis of neural networks, while providing \textit{exact quantitative} rather than \textit{statistical} guarantees \cite{wicker20probsafety}.

Another line of related works considers deriving exact or approximate abstractions of neural networks,
 which are applied for explanation~\cite{sotoudeh2021syrenn}, verification~\cite{elboher2020abstraction,pulina2010abstraction}, reachability analysis~\cite{prabhakar2019abstraction}, and preimage approximation~\cite{Dathathri19Inverse,ProveBound}.
\cite{Dathathri19Inverse} leverages symbolic interpolants \cite{Albarghouthi13interpolants} for preimage approximations, facing exponential complexity in the number of hidden neurons. 
Concurrently,
\cite{ProveBound} employs
Lagrangian dual optimization for preimage over-approximations.
Our anytime algorithm, which combines convex relaxation with principled splitting strategies for refinement, is applicable for both under- and over-approximations. 
Their work may benefit from our splitting strategies to scale to higher dimensions.

\section{Conclusion}
\label{sec:conclusion}
We present an efficient and flexible algorithm for preimage under-approximation of neural networks. 
Our \emph{anytime} method derives from the observation that linear relaxation can be used to efficiently produce under-approximations, in conjunction with custom-designed strategies for iteratively decomposing the problem to rapidly improve the approximation quality. Unlike previous 
approaches, it
is designed for, and scales to, both low and high-dimensional problems. Experimental evaluation on a range of benchmark tasks shows significant advantage in runtime efficiency and scalability, and the utility of our method for important applications in quantitative verification and robustness analysis.

\subsubsection{Acknowledgments}
This project received funding from the ERC under the European Union’s Horizon 2020 research and innovation programme (FUN2MODEL, grant agreement No.~834115) and ELSA: European Lighthouse on Secure and Safe AI project (grant agreement
No. 101070617 under UK guarantee). This
work was done in part while Benjie Wang was visiting the Simons Institute for the Theory of Computing.

\bibliographystyle{splncs04}
\bibliography{references}

\newpage
\appendix
\begin{center}
{\Large{\bf{Appendix}}}
\end{center}

\section{Experiment Setup}
We implement our preimage approximation approach and perform evaluation in Python. 
We use CROWN~\cite{crown} to perform linear relaxation on the nonlinear activation functions and bound propagation. 
We use the NNet package~\cite{nnet} and ONNX library of PyTorch \cite{pytorch} 
for neural network format conversion. 
For polyhedral operation and computation, we use the Parma Polyhedra Library (PPL) \cite{Bagnara2008PPL} to build the polyhedron object and use Qhull \cite{qhull} to compute the polytope volume.
All experiments are conducted on a cluster with Intel Xeon Gold 6252 2.1GHz CPU, and NVIDIA 2080Ti GPU.

\section{Experiment Evaluation}\label{sec:complete_eval}
In this section, we present the detailed configuration of neural networks and the complete experimental results on the benchmark tasks.
\subsection{Vehicle Parking.} 
For the vehicle parking task,
we consider computing the preimage approximation with input region corresponding to the entire input space
$\indomain = \{\inpoint \in \mathbb{R}^2 | \inpoint \in [0,2]^2 \}$, and output sets $\outset_{k}$, which correspond to the neural network outputting label $k$:

\begin{equation}
\outset_k = \{\outpoint \in \mathbb{R}^4~| \bigwedge_{i \in \{1, 2, 3, 4\}\setminus k} \outpoint_k - \outpoint_i \geq 0\},  \quad k \in \{1, 2, 3, 4\}
\end{equation}

We set the target coverage ratio as 90\% and the iteration limit of our algorithm as 1000.
The experimental comparison of the \textit{exact} preimage generation method and our approach are shown in Table~\ref{tab:parking}. 
From the evaluation results, we can see that our approach can effectively and efficiently generate preimage under-approximations with an average coverage ratio of 93.7\% and runtime cost of 0.722s.
The size of 
the polytope union is less than 4 for all output properties.
Compared with the polytope union generated  
by the exact solution,
our approach realizes an average reduction by 73.1\%.
As for the time cost, the exact method requires about 50 minutes for each output specification. In comparison, our approach demonstrates significant improvement
in runtime efficiency.
\begin{table}[tb]
\caption{Performance of preimage generation for vehicle parking.}\label{tab:parking}
\centering
\begin{tabular}{c|cc|ccc} 
 \toprule
        \multirow{2}{*}{\textbf{Property}} 
        & \multicolumn{2}{c|}{\textbf{Exact}}
         & \multicolumn{3}{c}{\textbf{Our}} \\
\cmidrule{2-6}
        &\textbf{\#Polytope} & \textbf{Time(s)}
         &\textbf{\#Polytope} & \textbf{Time(s)}
        & \textbf{PolyCov(\%)}   \\ 
      
\midrule
 $\{y\in \mathbb{R}^4| \wedge_{i \in \{2,3,4\}} y_1 \ge y_i\}$ 
& 10 & 3110.979
& 4  & 1.175 & 95.7\\ 
 $\{y\in \mathbb{R}^4| \wedge_{i \in \{1,3,4\}} y_2 \ge y_i\}$ 
& 20 & 3196.561
& 4  & 0.661 & 91.3\\ 
 $\{y\in \mathbb{R}^4| \wedge_{i \in \{1,2,4\}} y_3 \ge y_i\}$ 
& 7 & 3184.298
& 3  & 0.515 & 96.0 \\ 
 $\{y\in \mathbb{R}^4| \wedge_{i \in \{1,2,3\}} y_4 \ge y_i\}$ 
& 15 & 3206.998
& 3  & 0.535 & 91.9\\ 
\midrule
\textbf{Average}
& 13 & 3174.709 & 4 & 0.722 & 93.7\\
\bottomrule
\end{tabular}
\end{table}

\subsection{Aircraft Collision Avoidance}
\begin{table}[tb]
\caption{Performance of preimage generation for VCAS.}\label{tab:vcas}
\centering
\begin{tabular}{c|c c|c c c} 
\toprule
\multirow{2}{*}{\textbf{Models}}  
        & \multicolumn{2}{c|}{\textbf{Exact}}
         & \multicolumn{3}{c}{\textbf{Our}} \\
\cmidrule{2-6}
         & \textbf{\#Polytope} & \textbf{Time(s)} 
        & \textbf{\#Polytope} & \textbf{Time(s)} & \textbf{PolyCov(\%)} \\ 
\midrule
VCAS 1 & 49 & 6348.937 & 13  & 13.619 & 90.5 \\ 
VCAS 2  & 120 & 6325.712 & 11  & 10.431 & 90.4\\ 
VCAS 3 & 253 & 6327.981 & 18  & 18.05 & 90.0\\ 
VCAS 4 & 165 & 6435.46 & 11  & 10.384 & 90.4\\ 
VCAS 5 & 122 & 6366.877 & 11  & 10.855 & 91.1\\ 
VCAS 6 & 162 & 6382.198 & 10  & 9.202 & 92.0\\ 
VCAS 7  & 62 & 6374.165 & 6  & 4.526 & 92.8\\ 
VCAS 8  & 120 & 6341.173 & 14  & 13.677 & 91.3\\ 
VCAS 9 & 125 & 6366.941 & 11  & 10.782 & 90.3\\ 
\midrule
\textbf{Average}
& 131 & 6363.272 & 12 & 11.281 & 91.0\\
\bottomrule
\end{tabular}
\end{table}

The aircraft collision avoidance (VCAS) system~ \cite{Julian19nncontrol}
is used
to provide advisory for collision avoidance between the ownship aircraft and the intruder.
VCAS uses four input features 
$(h, \dot{h_A},\dot{h_B}, t)$ 
representing the relative altitude of the aircrafts,
vertical climbing rates of the ownship and intruder aircrafts, respectively, and time to the loss of horizontal separation. 
VCAS is implemented by nine feed-forward neural networks built  with a hidden layer of 21 neurons and an output layer of nine neurons corresponding to different advisories.

In our experiment, we use the 
following
input 
region
for the ownship and intruder aircraft as in \cite{Matoba20Exact}:  
$h \in [-8000, 8000]$, $\dot{h_A} \in [-100, 100]$, $\dot{h_B}=30$, and $t\in [0, 40]$.
We consider the output property $\outset = \{y\in \mathbb{R}^9~|\wedge_{i \in [2,9]}~y_1 \ge y_i\}$ and 
generate the preimage approximation for
the VCAS neural networks.

The experimental results are summarized in Table~\ref{tab:vcas}.
We compare our method with exact preimage generation, showing the number of polytopes (\#Polytope) in the under-approximation and exact preimage, respectively, and time in seconds (Time(s)).
Column ``PolyCov(\%)'' shows the approximate coverage ratio 
of our approach when the algorithm terminates.
For all VCAS networks, our approach effectively generates the preimage under-approximations with the polytope number varying from 6 to 18.
Compared with 
the exact method,
our approach realizes an average reduction of 91.1\% (131 vs 12).
Further, the computation time of our approach for all neural networks is less than 20s, and
$564\times$ faster than the exact method on average.

\subsection{Neural Network Controllers \& MNIST Classification}
In this subsection, we summarize the detailed configuration for neural network controllers in reinforcement learning tasks and MNIST classification tasks.
\subsubsection{Cartpole.} The cartpole control problem considers balancing a pole atop a cart by controlling the movement of the cart.
The neural network controller has two hidden layers with 64 neurons, and uses four input variables 
representing the position and velocity of the cart, the angle and angular velocity of the pole. The controller outputs are pushing the cart left or right.
In the experiments, we set the following input region for the Cartpole task: (1) cart position $[0,1]$, (2) cart velocity $[0,2]$, (3) angle of the pole $[-0.2,0]$, and (4) angular velocity of the pole $[-2,0]$ (with varied feature length in the evaluation). We consider the output property for the action \textit{pushing left}.

 \subsubsection{Lunarlander.} The Lunarlander problem considers the task of correct landing of  a moon lander on a landing pad. 
 The neural network for Lunarlander has two hidden layers with 64 neurons, and eight input features addressing the lander's coordinate, orientation, velocities, and ground contact indicators. The outputs represent four actions.
 For the Lunarlander task, we set the input region as: (1) horizontal and vertical position $[-1, 0]$ $\times$ $[0,1]$, (2) horizontal and vertical velocity $[0,2]$ $\times$ $[-2,0]$ (with varied feature length for evaluation), (3) angle and angular velocity $[-1,0] \times [-0.1,0.1]$, (4) left and right leg contact $[0.5,1]^2$. We consider the output specification for the action ``fire main engine'', i.e., $\{y \in \mathbb{R}^4~| \wedge_{i \in \{1,3,4\}} y_2 \ge y_i\}$.
\subsubsection{Dubinsrejoin.} The Dubinsrejoin problem considers guiding a wingman craft to a certain radius around a lead aircraft. 
The neural network controller has two hidden layers with 256 neurons.
The input space of the neural network controller is eight dimensional, with the input variables capturing the position, heading, velocity of the lead and wingman crafts, respectively.
The outputs are also eight dimensional representing controlling actions of the wingman. 
Note that the eight network outputs are processed further as tuples of actuators (rudder, throttle) for controlling the wingman where each actuator has 4 options.
The control action tuple is decided by taking the action with the maximum output value among the first four network outputs (the first actuator options) and the action with the maximum value among the second four network outputs (the second actuator options).
In the experiments, we set the following input region: 
(1) horizontal and vertical position $[-0.2, 0]\times [0, 0.5]$, (2) heading and velocity $[-1,0]\times [0, 0.2]$ for the lead aircraft,
and (3) horizontal and vertical position $[0.4, 0.6]\times [-0.3,0.3]$ (with varied feature length for evaluation), (4) heading and velocity $[0.2,0.5] \times [-1.5, 0]$ for the wingman aircraft.
We consider the output property that  both actuators (rudder, throttle) take the first option, i.e., $\{y\in \mathbb{R}^8~|\wedge_{i \in \{2, 3, 4\}} ~ y_1 \ge y_i \bigwedge \wedge_{i \in \{6, 7, 8\}}~y_5 \ge y_i\}$. 

\subsubsection{MNIST Classification.}
We use the trained neural network from VNN-COMP 2022 \cite{vnn22} for digit image classification. 
The neural network has six layers with a hidden neuron size of 100 for each hidden layer.
We consider
two types of image attacks: $l_{\infty}$ and patch attack.
For $L_{\infty}$ attack, a perturbation is applied to all pixels of the image.
For the patch attack, it applies arbitrary perturbations to the patch area, i.e., the perturbation noise covers the whole valid range $[0,1]$, for which we set the patch area at the center and (upper-left) corner of the image with different sizes.

\section{Effect of Methodological Components and Parameter Configurations}
\label{sec:ablation}

In this section, we perform an ablation study on two methodological components, subregion prioritization and input feature selection, 
and investigate the impact of two parameters in our 
approach, the sample suite size $N$ for approximation quality evaluation and the targeted coverage ratio for the tradeoff between approximation quality and efficiency.

\subsection{Ablation Study}
In this subsection, we examine the effectiveness of our subregion and input feature selection methods on the approximation quality.

\subsubsection{Subregion Selection.}
To evaluate the impact of our region selection method, we conduct comparison experiments with random region selection.
The random strategy differs from our approach in selecting the next region to split randomly without prioritization. 
We perform comparison experiments on the benchmark tasks.
Table~\ref{tab:ablation} summarizes the evaluation results of random strategy (Column ``Rand'') and our method (Column ``Our'')
for under-approximation refinement.
We set the same target coverage ratio and iteration limit for both strategies.
Note that for \textit{Dubinsrejoin}, random selection method hits the iteration limit and fails to reach the target coverage ratio.
The results confirm the effectiveness of our region selection method in that fewer iterations of the approximation refinement are required to reach the target coverage ratio, leading to
(1) a smaller number of polytopes (\#Polytope), reduced by 78.7\% on average,
and
(2)
a 75.5\% average runtime reduction. 

\begin{table}[t]
    \caption{Ablation study on methodological components.}
    \label{tab:ablation}
    \centering
    \begin{tabular}{cccc|ccc|ccc}
        \toprule
        \multirow{2}{*}{\textbf{Model}}
        &\multicolumn{3}{c}{\textbf{\#Polytope}}
         & \multicolumn{3}{c}{\textbf{PolyCov(\%)}}
         & \multicolumn{3}{c}{\textbf{Time(s)}} \\
        \cmidrule{2-10}
        & Rand
        & Heuristic
         & Our
         & Rand
          & Heuristic
         & Our
         & Rand
          & Heuristic
         & Our \\
        \midrule
        Vehicle Parking
         & 41
         & 4
         & 4
        & 90.5
        & 95.2
         & 91.3
         & 8.699
         & 0.967
         & 0.661
        \\
        \midrule
        VCAS
         & 56
         & 72
         & 13
        & 92.3
        & 90.0
         & 90.5
         & 29.15
         & 30.829
         & 13.619
        \\
        \midrule
        Cartpole
         & 151
         & 34
         & 32
        & 75.3
        & 75.3
         & 75.4
         & 84.639
         & 12.313
         & 17.156
        \\
        \midrule
        Lunarlander
         & 481
         & 238
         & 130
        & 75.1
        & 75.1
         & 75.0
         & 505.747
         & 119.744
         & 154.882
        \\
        \midrule
        Dubinsrejoin
         & 502
         & 105
         & 83
        & 61.8
        & 75.3
         & 75.1
         & 587.908
         & 65.869
         & 111.968
        \\
        \midrule
        \textbf{Average}
        & 246
        & 91
        & 52
        & 79.0
        & 82.2
         & 81.5
         & 243.229
         & 45.944
         & 59.657
        \\
        \bottomrule
    \end{tabular}
\end{table}

\subsubsection{Splitting Feature Selection.}
We compare our greedy
splitting method
with the heuristic splitting method 
which chooses to split the selected subregion along the input index with the largest value range. 
We present the comparison results of our method with the heuristic method (Column ``Heuristic'') in Table~\ref{tab:ablation}.
Our method requires splitting on all input features, computing the preimage approximations for all splits, and then choosing the dimension that refines the under-approximation the most, and we find that, even with parallelization of the computation over input features, our approach 
leads to
larger runtime overhead per-iteration
compared with the heuristic method.
Despite this, we find that our strategy actually requires \textit{fewer} refinement iterations to reach the target coverage, leading to a smaller number of polytopes (42.2\%  reduction on average) for the same approximation quality, demonstrating the per-iteration improvement in volume of the greedy vs heuristic strategy.

\subsection{Effect of Parameter Configurations}

\subsubsection{Sample Suite Size.}
In our approach, the sample suite size $N$ 
is used to estimate the volume of the polytope approximations and restricted preimage, which is then used to select the subregion to split and the feature to split on, as well as in the $\nnslope$-optimization procedure.
Figure~\ref{fig:sample_size} shows the evaluation results of the impact of sample suite size on preimage approximation for Cartpole. 
As demonstrated in Figure~\ref{fig:sample_size}, the number of polytopes (\#Polytope) that comprises the under-approximation 
decreases as the sample suite size increases until it reaches a fairly stable number (with minor fluctuation) with a sample size of 10000, with diminishing returns beyond that point.
This reflects that a smaller sample suite size leads to greater estimation error of the volumes
which negatively impacts both the refinement procedure and the $\nnslope$-optimization, 
resulting in more iterations needed to reach the target coverage.
On the other hand, a larger sample size provides more accurate estimation of the approximation quality, but requires larger runtime overhead, as reflected in the increase of runtime cost when the sample suite size continues to increase. 
In our experiments, we use a sample suite size of 10000 to balance between the estimation accuracy and runtime efficiency.

\begin{figure}[t]
\centering
\includegraphics[width=0.4\textwidth]{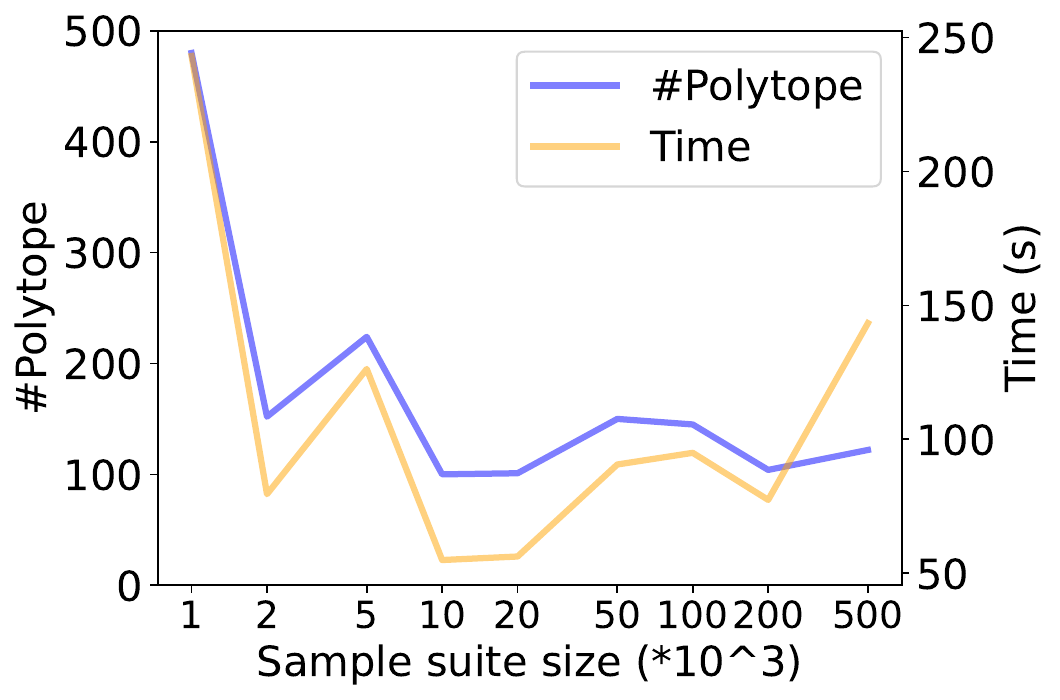}
\caption{Effect of sample suite size on preimage approximation.} \label{fig:sample_size}
\end{figure}

\subsubsection{Target Coverage.}
The target coverage controls the preimage approximation quality of our approach. 
Figure~\ref{fig:target_cov} shows the evaluation results of the impact of target coverage ratio on the polytope number
and runtime cost.
Overall, computing preimage approximation with larger target coverage ratio (better approximation quality) requires more refinement iterations on the input region, thus leading to more polytopes in the DUP approximation and larger runtime cost.
Our approach tries to minimize the polytope number 
to reach the target coverage ratio by prioritizing refinement on the input subregion with the lowest polytope coverage and selecting the splitting feature that improves the approximation quality the most.

\begin{figure}[t]
    \centering
    \begin{subfigure}{0.47\textwidth}
        \centering
        \includegraphics[width=\textwidth]{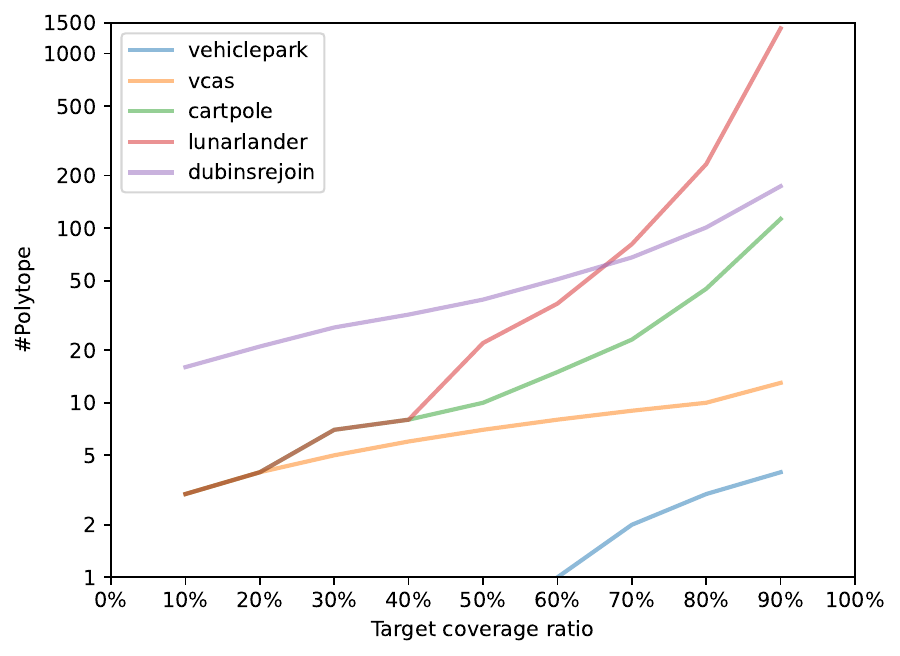}
        \caption{Polytope number}
        \label{fig:cov_subdm}
    \end{subfigure}
    \begin{subfigure}{0.47\textwidth}
        \centering
        \includegraphics[width=\textwidth]{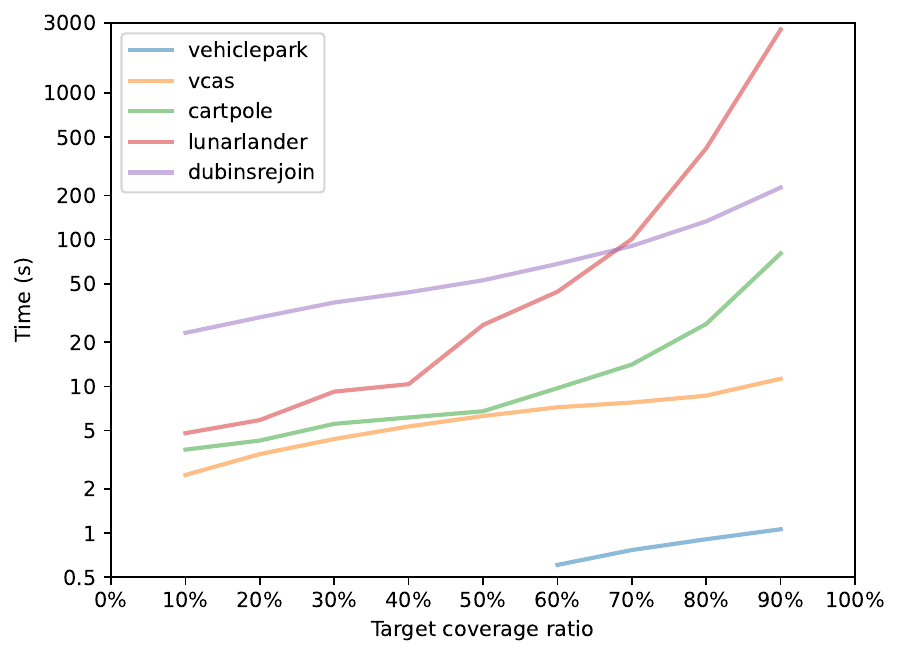}
        \caption{Runtime cost}
        \label{fig:cov_time}
    \end{subfigure}
    \caption{Effect of target coverage on preimage approximation.}
    \label{fig:target_cov}
\end{figure}

\section{Quantitative Verification for General Input Polytopes}\label{sec:generalpoly}

In Section \ref{sec:verif}, we detailed how to use our preimage under-approximation method to verify quantitative properties $(\inset, \outset, \proportion)$, where $\inset$ is a hyperrectangle. We now show how to extend this to a general polytope $\inset$.

Firstly, in Line \ref{algline:outerbox} of Algorithm \ref{alg:verify},  we derive a hyperrectangle $\indomain$ such that $\inset \subseteq \indomain$, by converting the polytope $\inset$ into its \textit{V-representation} \cite{grunbaum2003polytope}, which is a list of the vertices (extreme points) of the polytope, which can be computed as in \cite{avis1991pivoting,Barber96Qhull}. The computational expense of this operation is no greater than that of polytope exact volume computation. Once we have a V-representation, obtaining a bounding box can be achieved simply by computing the minimum and maximum value $\underline{\inpoint_i}, \overline{\inpoint_i}$ of each dimension among all vertices.  

Once we have the input region $\indomain$, we can then run the preimage refinement as usual, but with the modification that, when defining the polytopes and restricted preimages, we must additionally add the polytope constraints from $\inset$. Practically, this means that during every call to \texttt{EstimateVolume} and \texttt{ExactVolume} in Algorithm \ref{alg:verify}, we add these polytope constraints, and in Line \ref{algline:append} of Algorithm \ref{alg:genunderapprox}, we add the polytope constraints from $\inset$, in addition to those derived from the output $\outset$ and the box constraints from $\indomain_{sub}$.

\section{Proofs}\label{sec:proofs}

\propUnder*

\begin{proof}
The approximating polytope is defined as:
\begin{equation*}
    \polytope_{\indomain}(\outset) = \{\inpoint \in \mathbb{R}^{\indim}| \bigwedge_{i =1}^{\specnum} (\underline{\spec_i}(\inpoint)  \geq 0 )\wedge (\inpoint \in \indomain) \}
\end{equation*}  
while the restricted preimage is
\begin{equation*}
    \preimage_{\indomain}(\outset) = \{\inpoint \in \mathbb{R}^{\indim}| \nn(\inpoint) \in \outset \wedge \inpoint \in \indomain\} = \{\inpoint \in \mathbb{R}^{\indim}| \bigwedge_{i =1}^{\specnum} (\spec_i(\inpoint)  \geq 0 )\wedge (\inpoint \in \indomain) \}
\end{equation*}
\end{proof}

The LiRPA bound $\underline{\spec_i(\inpoint)} \leq \spec_i(\inpoint)$ holds for any $\inpoint \in \indomain$, and so we have $\bigwedge_{i =1}^{\specnum} (\underline{\spec_i}(\inpoint)  \geq 0 )\wedge \inpoint \in \indomain \implies \bigwedge_{i =1}^{\specnum} (\spec_i(\inpoint)  \geq 0 )\wedge \inpoint \in \indomain$, i.e. the polytope is an under-approximation.

\propPriority*

\begin{proof}
    We define $\polytope_{\indomain_{sub}}(\outset)|_{l}, \polytope_{\indomain_{sub}}(\outset)|_{r}$ to be the restrictions of $\polytope_{\indomain_{sub}}(\outset)$ to $\indomain_{sub}^l$ and $\indomain_{sub}^r$ respectively, that is:

    \begin{equation}
        \polytope_{\indomain_{sub}}(\outset)|_{l} = \{\inpoint \in \mathbb{R}^{\indim}| \bigwedge_{i =1}^{\specnum} (\underline{\spec_{i}}(\inpoint)  \geq 0 )\wedge (\inpoint \in \indomain_{sub}^l) \}
    \end{equation}
    \begin{equation}
        \polytope_{\indomain_{sub}}(\outset)|_{r} = \{\inpoint \in \mathbb{R}^{\indim}| \bigwedge_{i =1}^{\specnum} (\underline{\spec_{i}}(\inpoint)  \geq 0 )\wedge (\inpoint \in \indomain_{sub}^r) \}
    \end{equation}
    where we have replaced the constraint $\inpoint \in \indomain_{sub}$ with $\inpoint \in \indomain_{sub}^l$ (resp. $\inpoint \in \indomain_{sub}^r$), and $\underline{\spec_{i}}(\inpoint)$ is the LiRPA lower bound for the $i^{\textnormal{th}}$ specification on the input region $\indomain_{sub}$.

    On the other hand, we also have:
    \begin{equation}
        \polytope_{\indomain_{sub}^l}(\outset) = \{\inpoint \in \mathbb{R}^{\indim}| \bigwedge_{i =1}^{\specnum} (\underline{\spec_{l, i}}(\inpoint)  \geq 0 )\wedge (\inpoint \in \indomain_{sub}^l) \}
    \end{equation}
    \begin{equation}
        \polytope_{\indomain_{sub}^r}(\outset) = \{\inpoint \in \mathbb{R}^{\indim}| \bigwedge_{i =1}^{\specnum} (\underline{\spec_{r, i}}(\inpoint)  \geq 0 )\wedge (\inpoint \in \indomain_{sub}^r) \}
    \end{equation}
    where $\underline{\spec_{l, i}}(\inpoint)$ (resp. $\underline{\spec_{r, i}}(\inpoint)$) is the LiRPA lower bound for the $i^{\textnormal{th}}$ specification on the input region $\indomain_{sub}^l$ (resp. $\indomain_{sub}^r$). Now, 
    it is sufficient to show that $\polytope_{\indomain_{sub}^l}(\outset) \supseteq \polytope_{\indomain_{sub}}(\outset)|_{l}$ and $\polytope_{\indomain_{sub}^r}(\outset) \supseteq \polytope_{\indomain_{sub}}(\outset)|_{r}$ to prove Equation \ref{eqn:no_fragmentation}. We will now show that $\polytope_{\indomain_{sub}^l}(\outset) \supseteq \polytope_{\indomain_{sub}}(\outset)|_{l}$ (the proof for $\polytope_{\indomain_{sub}^r}(\outset) \supseteq \polytope_{\indomain_{sub}}(\outset)|_{r}$ is entirely similar).

    Before proving this result in full, we outline the approach and a sketch proof. It suffices to prove (for all $i$) that $\underline{\spec_{l, i}}(\inpoint)$ is a tighter bound than $\underline{\spec_{i}}(\inpoint)$ on $\indomain_{sub}^l$. That is, to show that $\underline{\spec_{l, i}}(\inpoint) \geq \underline{\spec_{i}}(\inpoint)$ for inputs $\inpoint$ in $\indomain_{sub}^l$, as then $\underline{\spec_{i}}(\inpoint)  \geq 0 \implies \underline{\spec_{l, i}}(\inpoint)  \geq 0$ for inputs $\inpoint$ in $\indomain_{sub}^l$, and so $\polytope_{\indomain_{sub}^l}(\outset) \supseteq \polytope_{\indomain_{sub}}(\outset)|_{l}$. The bound $\underline{\spec_{l, i}}(\inpoint)$ is tighter than $\underline{\spec_{i}}(\inpoint)$ because the input region for LiRPA is smaller for $\underline{\spec_{l, i}}(\inpoint)$, leading to tighter concrete neuron bounds, and thus tighter bound propagation through each layer of the neural network $\spec_i$.
    We present the formal proof of greater bound tightness for input and ReLU splitting in the following.

   \textbf{Input split:}
    We show $\underline{\spec_{l, i}}(\inpoint) \geq \underline{\spec_{i}}(\inpoint)$ for all $\inpoint \in \indomain_{sub}^l$ by induction (dropping the index $i$ in the following as it is not important). Recall that LiRPA generates symbolic upper and lower bounds on the pre-activation values of each layer in terms of the input (i.e. treating that layer as output), which can then be converted into concrete bounds.
    \begin{equation} \label{eqn:layer_bound}
        \lowerweight^{(j)} \inpoint + \lowerbias^{(j)} \leq \preact^{(j)}(\inpoint) \leq \upperweight^{(j)} \inpoint + \upperbias^{(j)}
    \end{equation}
    \begin{equation}  \label{eqn:layer_bound_left}
        \lowerweight^{(l, j)} \inpoint + \lowerbias^{(l, j)} \leq \preact^{(j)}(\inpoint) \leq \upperweight^{(l, j)} \inpoint + \upperbias^{(l, j)}
    \end{equation}
    where $\preact^{(j)}(\inpoint)$ are the pre-activation values for the $j^{\textnormal{th}}$ layer of the network $\spec_i$, and $\lowerweight^{(j)}, \lowerbias^{(j)}, \upperweight^{(j)}, \upperbias^{(j)}$ (resp. $\lowerweight^{(l, j)}, \lowerbias^{(l, j)}, \upperweight^{(l, j)}, \upperbias^{(l, j)}$) are the linear bound coefficients, for input regions $\indomain_{sub}$ (resp. $\indomain_{sub}^l$). 

    \textit{Inductive Hypothesis} For all layers $j = 1, ..., \numlayers$ in the network, and for all $\inpoint \in \indomain_{sub}^l$, it holds that:
    \begin{equation} 
    \lowerweight^{(j)} \inpoint + \lowerbias^{(j)} \leq
    \lowerweight^{(l, j)} \inpoint + \lowerbias^{(l, j)}  \leq \upperweight^{(l, j)} \inpoint + \upperbias^{(l, j)} \leq \upperweight^{(j)} \inpoint + \upperbias^{(j)}
    \end{equation}

    \textit{Base Case} For the input layer, we have the trivial bounds $\textbf{I}\inpoint \leq \inpoint \leq \textbf{I}\inpoint$ for both regions.

    \textit{Inductive Step} Suppose that the inductive hypothesis is true for layer $j - 1 < \numlayers$. Using the symbolic bounds in Equations \ref{eqn:layer_bound}, \ref{eqn:layer_bound_left}, we can derive concrete bounds $\concretelower^{(j - 1)} \leq \preact^{(j - 1)}(\inpoint) \leq \concreteupper^{(j - 1)}$ and $\concretelower^{(l, j - 1)} \leq \preact^{(l, j - 1)}(\inpoint) \leq \concreteupper^{(l, j - 1)}$ on the values of the pre-activation layer. By the inductive hypothesis, the bounds for region $\indomain_{sub}^l$ will be tighter, i.e. $\concretelower^{(j - 1)} \leq \concretelower^{(l, j - 1)} \leq \concreteupper^{(l, j - 1)} \leq \concreteupper^{(j - 1)}$. Now, consider the backward bounding procedure for layer $j$ as output. We begin by encoding the linear layer from post-activation layer $j - 1$ to pre-activation layer $j$ as:
    \begin{equation} \label{eqn:apx_linear}
        \weight^{(j)} \postact^{(j - 1)}(\inpoint) + \bias^{(j)} \leq \preact^{(j)}(\inpoint) \leq \weight^{(j)} \postact^{(j - 1)}(\inpoint) + \bias^{(j)}
    \end{equation}
    Then, we bound $\postact^{(j - 1)}(\inpoint)$ in terms of $\preact^{(j - 1)}(\inpoint)$ using linear relaxation. Consider the three cases in Figure \ref{fig:linear_relaxation_repeat} (reproduced from main paper), where we have a bound $\underline{c} \preact^{(j - 1)}_k(\inpoint) + \underline{d} \leq \postact^{(j - 1)}_k(\inpoint) \leq \overline{c} \preact^{(j - 1)}_k(\inpoint) + \overline{d}$, for some scalars $\underline{c}, \underline{d}, \overline{c}, \overline{d}$. If the concrete bounds (horizontal axis) are tightened, then an unstable neuron may become inactive or active, but not vice versa. It can thus be seen that the new linear upper and lower bounds on $\preact^{(j - 1)}_k(\inpoint)$ will also be tighter. 

    Substituting the linear relaxation bounds in Equation \ref{eqn:apx_linear} as in \cite{xu2021fast}, we obtain bounds of the form
    \begin{equation} 
        \lowerweight^{(j)}_j \preact^{(j - 1)}(\inpoint) + \lowerbias^{(j)}_j \leq \preact^{(j)}(\inpoint) \leq \upperweight^{(j)}_j \preact^{(j - 1)}(\inpoint) + \upperbias^{(j)}_j
    \end{equation}
    \begin{equation} 
        \lowerweight^{(l, j)}_j \preact^{(j - 1)}(\inpoint) + \lowerbias^{(l, j)}_j \leq \preact^{(j)}(\inpoint) \leq \upperweight^{(l, j)}_j \preact^{(j - 1)}(\inpoint) + \upperbias^{(l, j)}_j
    \end{equation}
    such that $\lowerweight^{(j)}_j \preact^{(j - 1)}(\inpoint) + \lowerbias^{(j)}_j \leq \lowerweight^{(l, j)}_j \preact^{(j - 1)}(\inpoint) + \lowerbias^{(l, j)}_j \leq \upperweight^{(l, j)}_j \preact^{(j - 1)}(\inpoint) + \upperbias^{(l, j)}_j \leq \upperweight^{(j)}_j \preact^{(j - 1)}(\inpoint) + \upperbias^{(j)}_j$ for all $\concretelower^{(l, j - 1)} \leq \preact^{(j - 1)}(\inpoint) \leq \concretelower^{(l, j - 1)}$, by the fact that the concrete bounds are tighter for $\indomain_{sub}^{l}$. 
    
    Finally, substituting the bounds in Equations \ref{eqn:layer_bound} and \ref{eqn:layer_bound_left} (for $\preact^{(j - 1)}$), and using the tightness result in the inductive hypothesis for $j - 1$, we obtain linear bounds for $\preact^{(j)}(\inpoint)$ in terms of of the input $\inpoint$, such that the inductive hypothesis for $j$ holds.

\textbf{ReLU split:}
We use $\indomain_{sub}^l$ and $\indomain_{sub}^r$ to denote the input subregions when fixing unstable ReLU neuron $z^{(j-1)}_k=\preact^{(j-1)}_k(\inpoint)$, i.e.,
$\indomain_{sub}^l=\{\inpoint~|~\preact^{(j-1)}_k(\inpoint) \ge 0\}$
and $\indomain_{sub}^r=\{\inpoint~|~\preact^{(j-1)}_k(\inpoint) < 0\}$.

In the following, we prove that $\underline{\spec_{l, i}}(\inpoint) \geq \underline{\spec_{i}}(\inpoint)$ for all $\inpoint \in \indomain_{sub}^l$.
Assume we fix one unstable ReLU neuron of layer $j-1$, then for all layers $1 \le m \le j-1$, for all   $\inpoint \in \indomain_{sub}^l$, it holds that:
\begin{equation} 
\lowerweight^{(m)} \inpoint + \lowerbias^{(m)} \leq
\lowerweight^{(l, m)} \inpoint + \lowerbias^{(l, m)}  \leq \upperweight^{(l, m)} \inpoint + \upperbias^{(l, m)} \leq \upperweight^{(m)} \inpoint + \upperbias^{(m)}
\end{equation}
where $\lowerweight^{(m)}=\lowerweight^{(l, m)}$, $\lowerbias^{(m)}=\lowerbias^{(l, m)}$ and same for the upper bounding parameters.

Now consider the bounding procedure for layer $j$.
The linear layer from post-activation layer $j - 1$ to pre-activation layer $j$ can be encoded as:
    \begin{equation}\label{eq:apx_linear_relu}
        \weight^{(j)} \postact^{(j - 1)}(\inpoint) + \bias^{(j)} \leq \preact^{(j)}(\inpoint) \leq \weight^{(j)} \postact^{(j - 1)}(\inpoint) + \bias^{(j)}
    \end{equation}
Consider the post activation function $\postact^{(j - 1)}(\inpoint)$ of the unstable neuron $z^{(j-1)}_k$, before splitting we have 
$\underline{c} \preact^{(j - 1)}_k(\inpoint) + \underline{d} \leq \postact^{(j - 1)}_k(\inpoint) \leq \overline{c} \preact^{(j - 1)}_k(\inpoint) + \overline{d}$, for some scalars $\underline{c}, \underline{d}, \overline{c}, \overline{d}$.
After splitting, we now have $\postact^{(j - 1)}_k(\inpoint) = \preact^{(j - 1)}_k(\inpoint)$ for $\indomain_{sub}^l$ where $\underline{c}=\overline{c}=1$, $\underline{d}=\overline{d}=0$, since the unstable neuron is fixed to be active.
By substituting the linear relaxation bounds before and after splitting in Equation \ref{eqn:apx_linear}, we obtain the bounding functions with regard to $\preact^{(j - 1)}(\inpoint)$ in the following form:
    \begin{equation} 
        \lowerweight^{(j)}_j \preact^{(j - 1)}(\inpoint) + \lowerbias^{(j)}_j \leq \preact^{(j)}(\inpoint) \leq \upperweight^{(j)}_j \preact^{(j - 1)}(\inpoint) + \upperbias^{(j)}_j
    \end{equation}
    \begin{equation} 
        \lowerweight^{(l, j)}_j \preact^{(j - 1)}(\inpoint) + \lowerbias^{(l, j)}_j \leq \preact^{(j)}(\inpoint) \leq \upperweight^{(l, j)}_j \preact^{(j - 1)}(\inpoint) + \upperbias^{(l, j)}_j
    \end{equation}
By the fact the relaxation is fixed to be exact for $\postact^{(j - 1)}_k(\inpoint)$, it holds that
$\lowerweight^{(j)}_j \preact^{(j - 1)}(\inpoint) + \lowerbias^{(j)}_j \leq \lowerweight^{(l, j)}_j \preact^{(j - 1)}(\inpoint) + \lowerbias^{(l, j)}_j \leq \upperweight^{(l, j)}_j \preact^{(j - 1)}(\inpoint) + \upperbias^{(l, j)}_j \leq \upperweight^{(j)}_j \preact^{(j - 1)}(\inpoint) + \upperbias^{(j)}_j$ for $\indomain_{sub}^{l}$.

Finally, for the bound propagation procedure of layer $L$, substituting the tightened bounding for $\preact^{(j-1)}(\inpoint)$, we obtain that $\underline{\spec_{l, i}}(\inpoint) = \lowerweight^{(l,L)} \inpoint + \lowerbias^{(l,L)}\geq \lowerweight^{(L)} \inpoint + \lowerbias^{(L)}=\underline{\spec_{i}}(\inpoint)$.
\end{proof}

\begin{figure}[t]
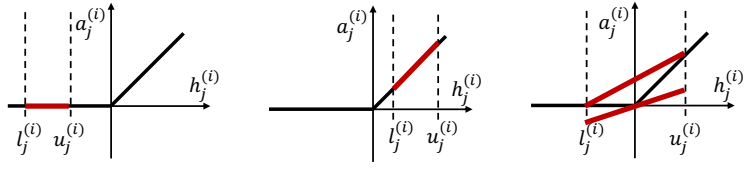

     \centering
     \begin{subfigure}[b]{0.25\textwidth}
         \centering
         \includegraphics[width=\textwidth]{Figs/negative.pdf}
         \label{fig:negative_app}
     \end{subfigure}
     \quad
     \begin{subfigure}[b]{0.25\textwidth}
         \centering
         \includegraphics[width=\textwidth]{Figs/positive.pdf}
         \label{fig:positive_app}
     \end{subfigure}
     \quad
     \begin{subfigure}[b]{0.25\textwidth}
         \centering
         \includegraphics[width=\textwidth]{Figs/unstable.pdf}
         \label{fig:unstable_app}
     \end{subfigure}
        \caption{Linear bounding functions for inactive, active, unstable ReLU neurons.}\label{fig:linear_relaxation_repeat}
\end{figure}

\corVolume*

\begin{proof}
    In each iteration of Algorithm \ref{alg:main}, we replace the polytope $\polytope_{\indomain_{sub}}(\outset)$ in a leaf subregion with two polytopes $\polytope_{\indomain_{sub}^l}(\outset), \polytope_{\indomain_{sub}^r}(\outset)$ in the DUP under-approximation. By Proposition \ref{prop:volume_guarantee}, the total volume of the two new polytopes is at least that of the removed polytope. Thus the volume of the DUP approximation does not decrease.
    
    Similarly, for ReLU splitting, we replace the polytope $\polytope_{\indomain_{sub}}(\outset)$ in a leaf subregion with two polytopes $\polytope_{\indomain_{sub}^l}(\outset), \polytope_{\indomain_{sub}^r}(\outset)$ where the relaxed bounding functions for one unstable neuron are replaced with exact linear functions, i.e., $\underline{c} \preact^{(i)}_j(\inpoint) + \underline{d} \leq \postact^{(i)}_j(\inpoint) \leq \overline{c} \preact^{(i)}_j(\inpoint) + \overline{d}$ is replaced with the exact linear function $\postact^{(i)}_j(\inpoint) = \preact^{(i)}_j(\inpoint)$ and $\postact^{(i)}_j(\inpoint) = 0$, respectively, as shown in Figure \ref{fig:linear_relaxation_repeat} (from unstable to stable).
     By Proposition \ref{prop:volume_guarantee}, the total volume of the two new polytopes is at least that of the removed polytope. Thus the volume of the DUP approximation does not decrease.
\end{proof}

\propSound*

\begin{proof}
    Algorithm \ref{alg:verify} outputs True only if, at some iteration, we have that the exact volume $\volume(\polytopeset) \geq \proportion \times \volume(\inset)$. Since $\polytopeset$ is an under-approximation to the restricted preimage $\preimage_{\inset}(\outset)$, we have that $\frac{\volume(\preimage_{\inset}(\outset))}{\volume(\inset)} \geq \frac{\volume(\polytopeset)}{\volume(\inset)} \geq \proportion$, i.e. the quantitative property $(\inset, \outset, \proportion)$ holds.
\end{proof}

\propComplete*
\begin{proof}
The proof for the soundness of Algorithm \ref{alg:verify} with ReLU splitting is similar to input splitting, with the main difference in deriving 
$\underline{\spec_{l, i}}(\inpoint) \geq \underline{\spec_{i}}(\inpoint)$ for all $\inpoint \in \indomain_{sub}$, which we have shown in the proof for Proposition \ref{prop:volume_guarantee}.

The proof for completeness is presented in the following.
Now that we have proved 
$\underline{\spec_{l, i}}(\inpoint) \geq \underline{\spec_{i}}(\inpoint)$ for all $\inpoint \in \indomain_{sub}^l$ and $\underline{\spec_{r, i}}(\inpoint) \geq \underline{\spec_{i}}(\inpoint)$ for all $\inpoint \in \indomain_{sub}^r$ after fixing a ReLU neuron $z^{(j)}_k$. 
When all unstable neurons are fixed with one activation status,
for each subregion $\indomain_{sub}$, we have $\underline{\spec_{i}}(\inpoint)=\spec_{i}(\inpoint)$.
Therefore, it holds that for any $\indomain_{sub} \subset \indomain$ where $\bigcup \indomain_{sub}=\indomain$, $(\underline{\spec_i}(\inpoint)  \geq 0 )\wedge \inpoint \in \indomain_{sub} \iff  (\spec_i(\inpoint)  \geq 0 )\wedge \inpoint \in \indomain_{sub}$, i.e., the polytope is the exact preimage.
Hence, when all unstable neurons are fixed to an activation status, we have $\polytopeset=\preimage_{\inset}(\outset)$.
Algorithm \ref{alg:verify} returns False only if 
the volume of the exact preimage 
$\frac{\volume(\preimage_{\inset}(\outset))}{\volume(\inset)} = \frac{\volume(\polytopeset)}{\volume(\inset)} < \proportion$.

\end{proof}

\paragraph{Discussion on $\nnslope$-optimization} It is worth noting that employing $\nnslope$-optimization breaks the guarantee in Proposition \ref{prop:volume_guarantee}, for two reasons. Firstly, due to the sampling and differentiable relaxation, the optimization objective is not exact volume. Secondly, gradient-based optimization cannot guarantee improvement in its objective after each update. Nonetheless, $\nnslope$-optimization can significantly improve the approximation in practice, and so we employ it in our method. 

\paragraph{Discussion on ReLU splitting} In the above, we have made the assumption that when we choose a neuron to split, all neurons in prior layers are stable over the subregion (for example, this occurs if we choose neurons to split in layer order). If this is not the case, the preactivation neuron $\preact^{(i)}_j(\inpoint)$ will not be a linear function of the input, and so we need to use linear lower/upper bounds to define the polytope under-approximations on the subregions generated by the split. This can also potentially break the guarantee in Proposition \ref{prop:volume_guarantee}. Note however that this does not affect completeness of the algorithm, as when all neurons are split, the preactivation neuron $\preact^{(i)}_j(\inpoint)$ will be a linear function of the input on every subregion.

\section{Preimage Over-Approximations} \label{sec:overapprox}

While we presented our method in the context of generating under-approximations to the preimage, it can be easily configured to instead generate over-approximations. The two key modifications are as follows.

Firstly, we need to be able to cheaply generate polytope over-approximations via linear relaxation. Analogously to Section \ref{sec:poly_gen}, we use LiRPA to propagate back a linear bound for each output constraint. However, instead of a linear lower bound, we use a linear \emph{upper} bound. More precisely, given a constraint $\lincon_i = (\linconw_i^{T} \outpoint + \linconb_i \geq 0)$, and the corresponding function $\spec_i(\inpoint) = \linconw_i^{T} \nn(\inpoint) + \linconb_i$,  we obtain upper bounds $\overline{\spec_i}(\inpoint) = \upperweightsingle_i^T \inpoint + \upperbiassingle_i$ for each $i$, such that $\spec_i(\inpoint) \geq 0 \implies \overline{\spec_i}(\inpoint) \geq 0 $ for $\inpoint \in \indomain$.  The polytope over-approximation is then given by $\polytope_{\indomain}(\outset) := \{\inpoint| \bigwedge_{i =1}^{\specnum} (\overline{\spec_i}(\inpoint)  \geq 0 )\wedge \bigwedge_{i = 1}^{\indim} \boxcon_i(\inpoint) \}$.

Secondly, instead of maximizing the volume of the DUP approximation, when over-approximating we wish to minimize the volume. This affects the following steps of our method:
\begin{itemize}
    \item For subregion selection (Section \ref{sec:branching}), we define the priority (Equation \ref{eqn:priority}) to instead be the volume of the over-approximation minus the true volume:
    \begin{equation*} 
    \textnormal{Priority}(\indomain_{sub})  = \volume(\polytope_{\indomain_{sub}}(\outset)) - \volume(\preimage_{\indomain_{sub}}(\outset)) 
    \end{equation*}
    The subregion with maximal priority then corresponds to the loosest over-approximation. 
    \item For local optimization (Section \ref{sec:local_opt}), we replace the lower bounds $\underline{\spec_i}$ with the upper bounds $\overline{\spec_i}$ constituting the over-approximation, and minimize rather than maximize the estimated volume $\widehat{\volume}(\polytope_{\indomain_{sub}}(\outset))$ with respect to $\nnslope$. 
\end{itemize}

\end{document}